\documentclass[onecolumn,11pt, draftcls]{IEEEtran}

\usepackage{amsmath,stackrel}
\usepackage{amssymb}
\usepackage{amsthm}
\usepackage{arydshln}
\usepackage{graphicx}

\usepackage{algorithm, algorithmic}

\usepackage{color}

\newtheorem{theorem}{Theorem}
\newtheorem{lemma}[theorem]{Lemma}
\newtheorem{definition}[theorem]{Definition}

\newtheorem{proposition}[theorem]{Proposition}

\newtheorem{corollary}[theorem]{Corollary}

\newcommand{\mypar}[1]{\vspace{0.1in}\noindent{\bf #1.}}
\DeclareMathAlphabet\mathbfcal{OMS}{cmsy}{b}{n}

\title{Detecting random walks on graphs with heterogeneous sensors}
\author{Dragana Bajovi\'c, Jos\'e M. F. Moura, Dejan Vukobratovi\'c}
\date{\today}

\begin{document}
\maketitle
%\displaydate
%

\mypar{Abstract}  We consider the problem of detecting a random walk on a graph, based on observations of the graph nodes. When visited by the walk, each node of the graph observes a signal of elevated mean, which we assume can be different across different nodes. Outside of the path of the walk, and also in its absence, nodes measure only noise. Assuming the Neyman-Pearson setting, our goal then is to characterize detection performance by computing the error exponent for the probability of a miss, under a constraint on the probability of false alarm. Since exact computation of the error exponent is known to be difficult, equivalent to computation of the Lyapunov exponent, we approximate its value by finding a tractable lower bound. The bound reveals an interesting detectability condition: the walk is detectable whenever the entropy of the walk is smaller than one half of the expected signal-to-noise ratio. We derive the bound by extending the notion of Markov types to Gauss-Markov types. These are sequences of state-observation pairs with a given number of node-to-node transition counts and the {\color{black} same average signal values across nodes, computed from the measurements made during the times the random walk was visiting each node's respective location}. The lower bound has an intuitive interpretation: among all Gauss-Markov types that are asymptotically feasible in the absence of the walk, the bound finds the most typical one under the presence of the walk. Finally, we show by a sequence of judicious problem reformulations that computing the bound reduces to solving a convex optimization problem, which is a result in its own right.

\mypar{Keywords} Random walk, hypothesis testing, error exponent, large deviations principle, threshold effect, Gauss-Markov type, convex analysis, Lyapunov exponent.

\section{Introduction}

Suppose we have a network of $N$ nodes, where each node is equipped with a sensor that measures the network environment. The environment can be in two states: 1) either a certain activity is present (e.g., an intruder, a signal), and the nodes have elevated mean; or 2) the environment is static, and the nodes measure only noise. We assume that the activity has the form of a random walk on the nodes of the graph, with a certain transition matrix $P$. We also assume that the measured signals are embedded in additive white Gaussian noise. The goal is to detect the random walk, based on the network nodes' observations.

This detection problem has widespread applicability. In~\cite{Hero06}, the authors consider the problem of detecting the spin of electrons using magnetic resonance force microscopy (MRFM). The observed signal is modeled as a random telegraph signal (i.e., Markov chain with states $0$ and $1$). Detecting an intruder by a sensor network (e.g., a video network) can also be modelled by this model. In~\cite{Candes06}, a similar, graphical model methodology is applied for detection of highly oscillatory signals (``chirps'').  In this paper, we present an application for random medium access in communications systems with extremely low signal-to-noise ratio (SNR) and unknown frequency selective fading. To establish communication with its associated base station, a user sends an access signal that hops from one frequency to another according to a Markov chain with a specified transition matrix; to detect a user, the base station then implements the corresponding likelihood ratio test (see eq.~\eqref{eq-LLR} further ahead). By performing frequency hopping, the effects of frequency selective fading are significantly alleviated, and furthermore without the need for synchronization between the sender and the receiver. This kind of scenario might be of interest for the emerging Narrow Band Internet of Things (IoT) standard, which envisions a similar random access setup for an extremely large number of communicating IoT devices and similar effects of unknown environments.

\mypar{Related literature} Detection of Markov chains hidden in noise has been considered in~\cite{Zig77},~\cite{Burnashev81},~\cite{Hero06},~\cite{Leong07}, and~\cite{Agaskar15}. For a more general review of hidden Markov processes, we refer the reader to the excellent overview paper~\cite{EphraimMerhav02}. {\color{black} In~\cite{Zig77} the author studies conditions for detectability of a random walk on a $d$-dimensional lattice of integers. It is shown that when $d=1$ or $d=2$, the random walk is always detectable, even for arbitrarily small values of signal to noise ratio (SNR). The reference also considers the general case $d\geq 3$ and finds a sufficient condition that guarantees that the walk is detectable and a necessary condition which, when violated, asserts that the walk cannot be detected. Under the same setup of a lattice random walk, reference~\cite{Burnashev81} shows that, for $d\geq 3$, there in fact exists a threshold on the SNR value which marks the border between the regions of detectability and undetectability. Although the existence of detectability threshold was proven, its exact computation remained an open question.} In~\cite{Hero06}, the authors consider spin detection and are concerned with deriving efficient detection tests and evaluating and comparing their performance. They show that, at any given time $t$, the optimal, likelihood ratio test can be conveniently expressed as a product $\Pi_t$ of $2t$ matrices, where the transition matrix $P$ alternates with independent and identically distributed (i.i.d.) diagonal matrices $D_t$ defined by network observations.

Papers closest to our work are~\cite{Leong07},~\cite{Agaskar15}. In~\cite{Leong07},~\cite{Agaskar15}, the authors consider Neyman-Pearson detection and study asymptotic performance of the likelihood ratio test as the number of observations per sensor grows. The assumed performance metric is the error exponent of the probability of a miss, under a fixed constraint on the probability of false alarm. Reference~\cite{Leong07} evaluates numerically the error exponent for a two-state Markov chain. Reference~\cite{Agaskar15} shows that finding the error exponent of the probability of a miss is equivalent to computing the Lyapunov exponent of the product $\Pi_t$ above -- a problem well-known to be difficult, see~\cite{TsitsiBlondel97}. Assuming identical SNR across nodes, the paper then finds a lower bound on the error exponent and shows by simulations that this bound is very close to the true error exponent value. {\color{black} The lower bound exhibits the same threshold value as the sufficient condition from~\cite{Zig77}: the walk is detectable if the SNR is greater than twice the entropy of the walk. While~\cite{Zig77} proves this result for a regular lattice (with equal transition probabilities at each node), reference~\cite{Agaskar15} generalizes this bound for general (although finite dimensional) random walks.}

In our previous work~\cite{Bajovic13}, we also studied products of i.i.d. random matrices. However, in contrast with the problem of evaluating the Lyapunov exponent, in~\cite{Bajovic13} we were concerned with evaluating the large deviations rate for the probability of the event that the product stays away from its limiting matrix (i.e., away from the Lyapunov exponent limit).

The setup that we study here is also related to random dynamical systems (iterated random functions)~\cite{Diaconis99}. In particular, the log likelihood ratio, see equation~\eqref{eq-zeta-via-Lyap-exp} in Section~\ref{sec-Setup}, can be represented  in the form of a random linear dynamical system, in which the random linear transformation has a specific form: a deterministic component of the dynamics -- the transition matrix $P$ -- is intertwined with a
random one -- the measurement dependent matrix $D_t$.

\mypar{Contributions} So far, random walk detection has only been considered under the assumption that the SNR values on the walk's path are equal. However, although convenient for analytical purposes, this assumption is often not realistic in practice. For example, in a video network, cameras closer to the intruder’s path will have better SNR than those further away. A communication signal that uses frequency hopping often experiences frequency selective fading -- an effect that the receiver must account for when designing signal detection tests. Motivated by these practical considerations, in this paper we address the scenario when the random walk signal is observed with different SNRs across nodes of the graph. We find a lower bound on the error exponent for Neyman-Pearson detection. {\color{black} The bound exhibits a threshold on the expected SNR value, thus generalizing the threshold of~\cite{Zig77} and~\cite{Agaskar15} to the case of unequal SNRs. Finally, we show that computing the error exponent lower bound is equivalent to solving a convex optimization problem, thus showing that it can be evaluated in a computationally efficient manner. The latter was not known even in the case of equal SNRs across nodes. We illustrate computation of the lower bound using the Frank-Wolfe (conditional gradient) algorithm~\cite{FrankWolfe56},~\cite{Jaggi13}. The bound is then compared with the true error exponent obtained by Monte Carlo simulations and the results show that the bound closely follows the true error exponent curve, under different simulation settings. }

To find the bound on the error exponent under the assumption of heterogeneous sensors, we had to follow a proof path different from~\cite{Agaskar15}. In particular, assuming that the SNRs across sensors are different, it is no longer possible to lump all measurements into a single quantity through summation (and similarly with the probability of the state sequence, see eq. (15) and the following text in~\cite{Agaskar15}). Instead, in our proofs, we increase the level of granularity by using the notion of Markov types -- sequences of states with the same transition counts, introduced by Davisson, Longo and Sgarro~\cite{Davisson81}. Following the intrinsic structure of the problem, where {\color{black} the random walk signal realizations are naturally grouped per each node where they have the same statistical properties, we extend the notion of Markov types to what we term Gauss-Markov types. The latter define pairs of state--signal sequences where the state sequences have equal Markov types and the  signal sequences have equal per-node average values}. We then prove the large deviations principle (LDP) for the {\color{black} empirical measure of the realization of the} Gauss-Markov type, when the state sequence is chosen uniformly at random from the set of all possible random walk sequences up to a given time. This result is the core technical result behind the lower bound on the error exponent.

\mypar{Paper organization} In Section~\ref{sec-Setup}, we pose the problem. In Section~\ref{sec-Gauss-Markov-types}, we introduce Gauss-Markov types, give preliminaries, and state the LDP result for Gauss-Markov types. In Section~\ref{sec-Lower-bound}, we state and prove Theorem~\ref{theorem-main}, the main result on the lower bound of the error exponent, while Section~\ref{sec-LDP-for-Gauss-Markov-proof} proves the LDP for Gauss-Markov types. Section~\ref{sec-Convexity} proves convexity of the error exponent lower bound and, as a by product, derives a solution by a single letter parametrization. {\color{black} Section~\ref{sec-Num-Results} illustrates evaluation of the bound and compares it with the true error exponent, and Section~\ref{sec-Conclusion} concludes the paper.}

\mypar{Notation}  For $N\in \mathbb N$, we denote by $1$ the vector of all ones in $\mathbb R^N$, and by $e_i$ the $i$-th canonical vector of $\mathbb R^N$ (that has value one on the $i$-th entry and the remaining entries are zero); we denote by $\mathbb R_{+}^N$ the set of vectors in $\mathbb R^N$ that have all elements non-negative {\color{black} and, similarly, by $\mathbb R_{+}^{N\times N}$ the set of all matrices in $\mathbb R^{N\times N}$ with nonnegative elements; for simplicity, sometimes $\mathbb R^{N\times N}$ is shortened to $\mathbb R^{N^2}$}. For a matrix $A$, we let $[A]_{ij}$ and $A_{ij}$ denote its $i,j$ entry and  for a vector $a\in \mathbb R^d$, we denote its $i$-th entry by $a_i$, $i,j=1,...,d$. For a function $f:\mathbb R^d\mapsto \mathbb R$, we denote its domain by $\mathcal D_f=\left\{ x\in \mathbb R^d: -\infty <f(x)<+\infty \right\}$; $\log$ denotes the natural logarithm and $\log_{+}$ denotes the function $\max\{0,\log\}$.  We let $\|\cdot\|$ denote the spectral norm of a square matrix. For $N$ real numbers $d_1,...,d_N$, we let $\mathrm{diag}\left\{d_1,...,d_N\right\}$ denote the diagonal matrix whose $i$th diagonal entry is $d_i$, for $i=1,...,N$.  {\color{black} A closed box in $\mathbb R^d$ of width $\rho$ and centered at $x$ is denoted by $B_{x}(\rho)$}; the closure, the interior, the boundary, and the complement of an arbitrary set $G\subseteq \mathbb R^d$ are respectively denoted by $\overline G$, $G^{\mathrm{o}}$, $\partial G$, and $G^{\mathrm{c}}$; $\mathcal B(\mathbb R^d)$ denotes the Borel sigma algebra on $\mathbb R^d$; {\color{black} $(\Omega,\mathcal F,\mathbb P)$ denotes a probability space, with sample space $\Omega$, sigma algebra $\mathcal F$, and probability measure $\mathbb P$, and we denote an outcome in $\Omega$ by $\omega$}; $\mathbb E$ denotes the expectation operator; $\mathcal N(m,S)$ denotes Gaussian distribution with mean vector $m$ and covariance matrix $S$; $\mathcal M(V,v_0,P)$ denotes a Markov chain on a finite set of states $V$ with initial distribution $v_0$ and a transition matrix $P$. For any integer $t$, $S_t$ denotes the state of the Markov chain at time $t$. For any integer $t$, $S^t$ denotes the sequence of the first $t$ states of the Markov chain, i.e., $S^t=(S_1,S_2,...,S_t)$.

\section{Problem setup}
\label{sec-Setup}
We consider testing the hypothesis whether or not there is an object (an agent) following a certain random walk on a given graph $G=(V,E)$ of $N$ nodes, where $V=\{1,2,...,N\}$ denote the set of nodes and $E$  denotes the set of edges of the graph. The transition matrix of the random walk is known and we denote it by $P$. We assume that $P$ is irreducible and aperiodic, and we denote the (unique) stationary distribution of the walk by {\color{black} $\pi>0$ (note that uniqueness and positivity follow from the Perron-Frobenius theorem for irreducible matrices, e.g.,~\cite{MatrixAnalysis})}. We also assume that the initial distribution is the stationary distribution, i.e., walk starts at node $i$ with probability $\pi_i$, for $i=1,...,N$.

At any time $t$, we denote by $S_t$ the node that the agent visits at time $t$ (the state of the Markov chain at time $t$). Each node in the graph $i$, at each time $t$, produces a noisy measurement of the activity at its location, which we denote by $X_{i,t}\in \mathbb R$. We assume that, in the absence of the walk, the measurement of node $i$ is standard normal (and hence the same for all nodes), and when visited by the walk, the measurement of node $i$ is  normally distributed with mean $\beta_i$ and variance equal to one.Thus, the SNR resulting from the presence of an agent performing the random walk (``activity'') is different across different nodes.  Summarizing, the two hypotheses we consider are:
\begin{align}
\label{def-hypothesis-testing}
\mathcal H_0:\; &X_{i,k}\stackrel{\mathrm{i.i.d.}}{\sim}\mathcal N(0,1)\\
\mathcal H_1:\; &X_{i,k}|S^t \stackrel{\mathrm{indep.}}{\sim}
\left\{\begin{array}[+]{ll}
\mathcal N(\beta_i,1), &\mathrm{\;if\;} S_k=i\\
\mathcal N(0,1), &\mathrm{\;if\;} S_k\neq i
\end{array}\right.,\\
& \phantom{X_{i,t}|S^t \stackrel{\mathrm{indep.}}{\sim}}\mathrm{where}\;\;\; S^t=(S_1,...,S_t)\sim \mathcal M(V,\pi, P),
\end{align}
for $i=1,2,...,N,$ $k=1,2,...,t$, and where, we recall, $\mathcal M(V,\pi,P)$ denotes a Markov chain on the set of states $V$, with initial distribution $\pi$ (the stationary distribution of the chain), and the transition matrix $P$.  For each $t$, we let the random matrix $\mathbfcal{X}^t\in \mathbb R^{N\times t}$ collect measurements of all nodes up to time $t$ in column-wise fashion, such that the $(i,k)$ entry of $\mathbfcal {X}^t$ stores the measurement of node $i$ at time $k$, i.e., $\mathbfcal {X}^t_{ik}=X_{i,k}$, for each $i$ and $k\leq t$.

We denote the probability laws corresponding to $\mathcal H_0$ and $\mathcal H_1$ by $\mathbb P_0$ and $\mathbb P_1$, respectively. Similarly, the expectations with respect to $\mathbb P_0$ and $\mathbb P_1$ are denoted by $\mathbb E_0$ and $\mathbb E_1$, respectively. The probability density functions of $\mathbfcal X^t$ under $\mathcal H_0$ and $\mathcal H_1$ are denoted by $f_{0,t}(\cdot)$ and $f_{1,t}(\cdot)$. It will also be of interest to introduce the conditional probability density function of $\mathbfcal X^t$ given $S^t=s^t$ (i.e., the likelihood functions{\color{black}~\cite{Zig77}},~\cite{EphraimMerhav02}), which we denote by $f_{1,t|S^t}(\cdot|s^t)$, for any $s^t$. Finally, the likelihood ratio at time $t$ is denoted by $L_t$ and at a given realization of $\mathbfcal X^t$ is computed by $L_t(\mathbfcal X^t)= \frac{f_{1,t}({\mathbfcal{X}}^t )}{f_{0,t}({\mathbfcal{X}}^t)}$.

\mypar{Error exponent} In this paper, we consider Neyman-Pearson hypothesis testing, and we are interested in computing the error exponent for the probability of a miss, given a threshold $\alpha$  on the probability of false alarm. For each $t$, let $P_{\mathrm{miss},t}^\alpha$ denote the infimum of {\color{black} the} probability of a miss among all tests such that the resulting probability of false alarm is below $\alpha$. {\color{black} Then, our goal is to compute, or characterize,
\begin{equation}\label{eq-error-exponent}
\lim_{t\rightarrow +\infty} - \frac{1}{t} \log P_{\mathrm{miss},t}^\alpha = \zeta,
\end{equation}
provided that the above limit exists. It is well known that, in the case of i.i.d. observations under both hypotheses, the above limit does exist, i.e., $P_{\mathrm{miss},t}^\alpha$ decays exponentially fast to zero, for any $\alpha>0$, and with the rate $\zeta$ equal to the Kullback-Leibler distance between $\mathbb P_0$ and $\mathbb P_1$. The latter result is known as the Stein's (or Chernoff-Stein's) lemma~\cite{Chernoff56},~\cite{Bahadur71},~\cite{DemboZeitouni93}. A generalization of the Stein's lemma for ergodic stochastic processes is presented in~\cite{Bahadur80}. This work is concerned with the limit of the scaled log-likelihood ratios in the almost sure sense:
\begin{equation}\label{eq-asymptotic-KL-rate}
\lim_{t\rightarrow +\infty} - \frac{1}{t} \log L_t(\mathbfcal X^t) = \kappa,
\end{equation}
and it asserts that when the above limit exists under $\mathbb P_0$, then the error exponent $\zeta$ also exists and moreover it holds that $\zeta=\kappa$~\cite{Bahadur80}{ \footnote{{\color{black}The quantity in~\eqref{eq-asymptotic-KL-rate} is sometimes referred to in the literature as the ``asymptotic KL rate'', see, e.g.,~\cite{Sung06}. In another line of works, concerned also with the existence of the error exponent in the ergodic case~\cite{Vajda89},~\cite{Vajda90} (see also~\cite{Luschgy93},~\cite{Chen96}), ``asymptotic KL rate'' is defined, not as the almost sure limit, but as the limit in expectation of the scaled log-likelihood ratios:
\begin{equation}
\label{eq-asymptotic-KL-rate-2}
\kappa^\prime=\lim_{t\rightarrow+\infty} - \frac{1}{t} \mathbb E_0\left[\log L_t(\mathbfcal X^t)\right],
\end{equation}
or, in other words, the limit of the (normalized) KL divergences across iterations. In~\cite{Vajda89},~\cite{Vajda90},~\cite{Chen96} it was shown that, if under $\mathbb P_0$ the scaled log-likelihood ratio converges with probability one to a limit,  then the limit in~\eqref{eq-error-exponent} exists and equals $\zeta= \kappa^\prime$. Finally, we remark that conditions for discriminating between two measures corresponding to arbitrary random sequences were studied in earlier works~\cite{Kakutani48} (the case of independent, but not identically distributed observations), and in~\cite{Skorokhod74},~\cite{Kabanov77}, but without the study of the rate of discrimination (the error exponent). In these works it was shown that the measures can be discriminated if and only if, under $\mathbb P_0$, the sequence of likelihood ratios $L_t$ converges with probability one to $0$. }}}}.

To compute the error exponent $\zeta$, we express the likelihood ratio in terms of the likelihood functions $f_{1,t|S^t}$~{\color{black}~\cite{Zig77}},~\cite{EphraimMerhav02}. For short, we denote (with some abuse of notation) $P(s^t)= \mathbb P_1(S^t=s^t)$, i.e., for any $s^t$, $P(s^t) = \pi_{s_1} \prod_{k=1}^{t-1} P_{s_{k}s_{k+1}}$. Further, for each $t$, let $\mathcal S^t$ denote the set of all feasible sequences $s^t$ of length $t$, i.e., $\mathcal S^t = \left\{s^t = (s_1,...,s_t): P(s^t)>0\right\}$, and let $C_t$ denote its cardinality, $C_t=|\mathcal S^t|$. By conditioning on the random walk realizations up to time $t$, it is easy to see that the likelihood ratio at time $t$ can be expressed as:
\begin{align}
\label{eq-LLR}
L_t(\mathbfcal{X}^t)&= \sum_{s^t \in \mathcal S^t} P(s^t)\frac{f_{1,t|S^t}(\mathbfcal {X}^t|s^t)}
{f_{0,t}(\mathbfcal {X}^t)}\nonumber\\
&=\sum_{s^t \in \mathcal S^t} P(s^t)e^{\sum_{k=1}^t \beta_{s_k} X_{s_k,k}- \frac{\beta_{s_k}^2}{2}}.
\end{align}

In~\cite{Hero06},~\cite{HeroICASSP06} the authors show that, seemingly combinatorial in nature, the sum in~\eqref{eq-LLR} can in fact be conveniently expressed as a matrix product:
\begin{equation}
\label{eq-zeta-via-Lyap-exp}
L_t(\mathbfcal{X}^t) =  {\color{black}\pi^\top} D_t P D_{t-1} P\ldots D_1 P 1,
\end{equation}
where, for each $t=1,2,...$, $D_t$ is a diagonal matrix defined by $D_t=\mathrm{diag} \left(e^{ \beta_1 X_{1,t} - \frac{\beta_1^2}{2}},...,e^{\beta_N X_{N,t} - \frac{\beta_N^2}{2}} \right)$, {\color{black} and where the measurements are taken under $\mathcal H_0$}. We remark that from~\eqref{eq-asymptotic-KL-rate} it follows that the error exponent $\zeta$, see~\eqref{eq-error-exponent}, is equal to the (top) Lyapunov exponent~\cite{Furstenberg1960}. Using the fact that the measurement {\color{black}vectors $\left(X_{1,t},...,X_{N,t}\right)$ are i.i.d. under $\mathcal H_0$, it follows that the matrices $D_t$ are also  i.i.d., and the existence of the limit in~\eqref{eq-asymptotic-KL-rate} follows} by the well-known Furstenberg-Kesten theorem~\cite{Furstenberg1960}. We formally state this result in Lemma~\ref{lemma-asymptotic-KL-rate}. Using the fact that $\pi>0$, the proof of Lemma~\ref{lemma-asymptotic-KL-rate} consists of verifying the condition of the Furstenberg-Kesten's theorem, i.e., proving that the expectation $\mathbb E_0\left[ \log_{+} \|PD_t\|\right]$ is finite.
{\color{black}
\begin{lemma}\label{lemma-asymptotic-KL-rate}
The limit in~\eqref{eq-asymptotic-KL-rate} exists almost surely, with respect to $\mathbb P_0$.
\end{lemma}
As a side remark, and as a byproduct we note that the application of the Furstenberg-Kesten's gives that the limit in~\eqref{eq-asymptotic-KL-rate} equals the limit of the normalized KL divergence computed across iterations $t=1,2,...$, i.e.,
\begin{equation}\label{eq-asymptotic-KL-rate-exp}
\lim_{t\rightarrow +\infty}  - \frac{1}{t} \log L_t(\mathbfcal X^t) = \lim_{t\rightarrow +\infty}  - \frac{1}{t} \mathbb E_0\left[\log L_t(\mathbfcal X^t)\right],
\end{equation}
where the latter holds almost surely with respect to $\mathbb P_0$. Thus, the two conditions for the generalized Stein's lemma referred to in the preceding text --  the almost sure limit (the left-hand side of~\eqref{eq-asymptotic-KL-rate-exp}) from~\cite{Bahadur80} and the limit in expectation (the right-hand side of ~\eqref{eq-asymptotic-KL-rate-exp}) from~\cite{Vajda89},~\cite{Vajda90}, are in our case equivalent.
}

The computation of the Lyapunov exponent is known to be a very difficult problem~\cite{TsitsiBlondel97}, even in the case when the sample space of random matrices consists of only two matrices~\cite{TsitsiBlondel97}. Thus, our goal is finding tractable upper and lower bounds for $\zeta$ (as computed from~\eqref{eq-asymptotic-KL-rate}). %Similarly as in~\cite{Agaskar15}, we base our analysis of the error exponent $\zeta$ on the expression~\eqref{eq-asymptotic-KL-rate-exp}.

\mypar{Upper bound for $\zeta$} We end this section with a simple and intuitive upper bound for $\zeta$. Suppose that we knew in advance the exact path $s^t$ that the random walk will take. {\color{black} With increasing probability, this path will be a typical one, implying that, when $t$ is large, the number of times that the random walk is in state $i$ along $s^t$ is very close to $\pi_i t$. } {\color{black} The likelihood ratio then equals $ L_t(\mathbfcal X^t)= e^{\sum_{k=1}^t \beta_{s_k} X_{s_k,k}- \frac{\beta_{s_k}^2}{2}}$} and the error exponent~\eqref{eq-asymptotic-KL-rate-exp} is computed by
\begin{align}
\label{eq-error-exponent-perfect-knowledge-of-the-path}
\lim_{t\rightarrow +\infty}  - \frac{1}{t} \mathbb E_0\left[ \sum_{k=1}^t \beta_{s_k} X_{s_k,k}- \frac{\beta_{s_k}^2}{2}\right] & = \lim_{t\rightarrow +\infty} \frac{1}{t} \sum_{k=1}^t \frac{\beta_{s_k}^2}{2}\\
%& = \lim_{t\rightarrow +\infty} \frac{1}{t} \mathbb E_1\left[\beta_{S_t}^2\right]
& = \sum_{i=1}^N \pi_i \frac{\beta_i^2}{2},
\end{align}
where the last equality follows by the typicality of $s^t$. Now, given that it operates with knowledge of the exact path of the random walk, it is intuitive to expect that the error exponent in~\eqref{eq-error-exponent-perfect-knowledge-of-the-path} will be an upper bound for the error exponent~\eqref{eq-error-exponent} for the random walk detection problem~\eqref{def-hypothesis-testing}. Proposition~\ref{prop-UB} which we present next formalizes mathematically this intuitive notion; the proof is given in Appendix~\ref{app-A}.
\begin{proposition}
\label{prop-UB}
  There holds $\zeta \leq \overline \zeta$, where
\begin{equation}
\label{eq-main-UB}
\overline \zeta = \sum_{i=1}^N \pi_i \frac{\beta_i^2}{2}.
\end{equation}
%which equals one half of the limiting value of the total SNR averaged over the random walk's path,  $\lim_{t\rightarrow +\infty} \frac{1}{2t}\mathbb E_1\left[\beta_{S_t}^2\right]$.
\end{proposition}

\section{Gauss-Markov types}
\label{sec-Gauss-Markov-types}
In this section, we review concepts and results from the literature that we used in our study of the limit in~\eqref{eq-asymptotic-KL-rate}. We also define some novel concepts that will prove instrumental in the analysis of this limit. Specifically, building on the notion of Markov types from~\cite{Davisson81}, we introduce the notion of Gauss-Markov types which, to each sequence of states $s^t$, besides Markov type, associates also the vector of {\color{black} the nodes' local average signal values, computed at each node from the measurements collected during the random walk's visits to their respective locations}. We then state the main result behind the lower bound on the error exponent, Theorem~\ref{theorem-Q-t-satisfies-LDP}, which asserts that the {\color{black} sequence of empirical measures of the realizations of the  Gauss-Markov type} satisfies the LDP. Given the importance of Theorem~\ref{theorem-Q-t-satisfies-LDP} in the analysis of the error exponent, we dedicate Section~\ref{sec-LDP-for-Gauss-Markov-proof} to its proof.

\mypar{Transition counts matrix and Markov types} For any given $t\geq 1$ and any given sequence $s^t\in V^t$, for each $i=1,...,N$, we denote by $K_{t,i}(s^t)$  the number of times $k$ along the sequence $s^t$ {\color{black} the chain visits $i$, i.e., such that} $s_k=i$, in other words $K_{t,i} (s^t)=|\left\{ 1\leq k\leq t: s_k=i\right\}|$. Similarly, for every pair $(i,j)$, we denote by ${[K_t]}_{ij} (s^t)$ the number of times along the sequence $s^t$ when the state switches from $i$ to $j$, i.e., ${[K_t]}_{ij} (s^t)=|\left\{1\leq k\leq  t: s_k=i,\,s_{k+1}=j\right\}|$, {\color{black}where for $k=t$, we take $s_{t+1}=s_1$ (intuitively speaking, the definition of $K_t$ ``sees'' the sequence $s^t$ as being circular -- the motivation for this is given the text further below)}~\footnote{{\color{black} This ``adjustment'' in the definition of $[K_t]_{ij}$ is not essential for our analysis and is made only to make certain arguments more elegant and also to simplify the notation. In particular, we use $[K_t]_{ij} $ in the proof of Lemma~\ref{lemma-main-proof-step-2}, Appendix~\ref{app-B}, to express the probability $P(s^t)$ of a given sequence $s^t$: $P(s^t)= \frac{\pi_{s_1}}{P_{s_t, s_1}} e^{\sum_{i,j=1}^N K_{ij} (s^t) \log P_{ij} }$, where we note that dividing by $P_{s_t, s_1}$ takes care of the extra count from $s_t$ to $s_1$ in the definition of $[K_t]_{ij}$, see also eq.~\eqref{eq-P-s-t-expressed} and the text following this equation in the proof of Lemma~\ref{lemma-main-proof-step-2} in Appendix~\ref{app-B}}}.   Matrix $K_t$ defined in this way is known as the transition counts matrix~\cite{Boza71}.
{\color{black}  Then, it is clear that for any $i$ such that $i\neq s_1$ and $i\neq s_{t}$, the number of occurrences of state $i$ along $s^t$ must be the same whether we counted times when the Markov chain enters or leaves this state. Consider now $i=s_1$. Looking at the times when the Markov chain enters state $i$, we see from the definition of $[K_t]_{ji}$ that we are accounting for this first occurrence of state $i$ by the ``virtual'' transition from $s_t$ (back) to $s_1$. A similar analysis applies to the case $i=s_t$. Summarizing, we have that ${K_{t,i}}\equiv \sum_{j=1}^N {[K_t]}_{ij}$, and also ${K_{t,i}}\equiv \sum_{j=1}^N {[K_t]}_{ji}$.}

\begin{definition}[Markov type~\cite{Davisson81}]
\label{def-Markov-type} Markov type on the set of states $V$ at time $t$ is the matrix $\Theta_t: V^t \mapsto \mathbb R^{N\times N}$, where, for each $s^t\in V^t$, $\Theta_t(s^t)$ is defined by
\begin{equation}\label{def-Markov-type}
{[\Theta_t]}_{ij}(s^t):= \frac{{[K_t]}_{ij}(s^t)}{t},
\end{equation}
for any $i,j=1,...,N$.
\end{definition}

For any fixed $t=1,2,...$, we define the set $\Delta_t$ that contains all possible Markov types\footnote{The set of all possible Markov types at time $t$, $\Delta^\star_t\subset \Delta_t$, compared to $\Delta_t$ has an extra condition in its definition: it requires that the submatrix obtained by deleting all zero rows and columns from the candidate matrix $\theta$, is irreducible, see~\cite{Natarajan85}. However, for our purposes this condition can be omitted, due to the fact that both $\cup_{t=1}^{+\infty}\Delta^\star_t$ and $\cup_{t=1}^{+\infty}\Delta_t$ are dense in $\Delta$ (defined further ahead in~\eqref{def-Delta}), see~\cite{Natarajan85}.} at time $t$:
\begin{align}
\label{def-Delta-t}
\Delta_{t}&=\left\{\theta \in \mathbb R^{N\times N}: \mathrm{for\; every\;}(i,j),
\mathrm{\;there\; exists\;} k_{ij}\in\mathbb Z \mathrm{\;s.t.\;}  \theta_{ij}=\frac{k_{ij}}{t},
\phantom{\sum_{i,j=1}^N k_{ij}=t}\right. \nonumber\\
&  \;\;\;\;\;\;\; \left. \mathrm{\;where\;} \sum_{i,j=1}^N k_{ij}=t,\, {\color{black} \sum_{j=1}^N k_{ij}=\sum_{j=1}^N k_{ji},\,\mathrm{for}\;i=1,...,N,\;}
k_{ij}\geq 0\;\mathrm{and}\;k_{ij}=0\;\mathrm{if}\;(i,j)\notin E\right\},
\end{align}
It will also be of interest to introduce the set $\Delta$ that contains the union of all $\Delta_t$ sets, $t\geq 1$:
\begin{align}
\label{def-Delta}
\Delta&=\left\{\theta \in \mathbb R^{N\times N}: \sum_{i,j=1}^N \theta_{ij}=1,{\color{black} \,\sum_{j=1}^N \theta_{ij}
=\sum_{j=1}^N \theta_{ji},\,\mathrm{for}\;i=1,...,N,\;} \theta_{ij}\geq 0\;\mathrm{and}\;\theta_{ij}=0\;\mathrm{if}\;(i,j)\notin E\right\}.
\end{align}
For a given $\theta\in \Delta$, we let $\overline \theta$ denote the vector of row sums of $\theta$,
$\overline \theta=\theta 1$; note that $\overline \theta \in \mathbb R_{+}^N$.

For each $\theta \in \Delta$, we define, for each $t\geq 1$, the set of sequences of length $t$ that have the same Markov type $\theta$:
\begin{equation}
\label{def-C-t-theta}
\mathcal S^t_\theta=\left\{s^t\in \mathcal S^t: K_{ij}(s^t) = \theta_{ij} t,\, i,j=1,...,N\right\}.
\end{equation}
We let $C_{t,\theta}$ denote its cardinality, $C_{t,\theta}= |\mathcal S^t_{\theta}|$. Note that $\mathcal S^t_{\theta}$ is non-empty if and only if $\theta \in \Delta_t$ (if $\theta \notin \Delta_t$, by the definition of $\Delta_t$, we have that the number of transition fractions $\theta$ is not realizable at time $t$, i.e., there is no sequence of length $t$ such that, for each $i,j$, the number of transitions from $i$ to $j$ equals $\theta_{ij}t$).

\mypar{Entropy functions} We now define relevant entropy functions~\cite{Cover06} (see also~\cite{Boza71},~\cite{Vasek80},~\cite{Natarajan85},~\cite{Davisson81}), that we will utilize in estimating the size of the sets $C_{t,\theta}$, $\theta \in \Delta_t$. To this end, note that each $\theta \in \Delta$ defines the respective Markov chain with transition matrix $Q$, defined by $Q_{ij}:=\theta_{ij}/\overline \theta_i$,  when $\overline \theta_i\neq 0$, and $Q_{ij}=0$, otherwise, for $i,j=1,...,N$. (We note in passing that if $\theta \in \Delta_t$, then $Q$ is the empirical transition matrix for the transition count matrix $K$ that validates the fact that $\theta$ belongs to $\Delta_t$, see eq.~\eqref{def-Delta-t}). {\color{black} Then, for each $\theta\in \mathbb R^{N\times N}_{+}$ we define:
\begin{equation}
\label{def-entropy}
H(\theta)= - \sum_{i,j=1}^N \theta_{ij}\log \frac{\theta_{ij}}{\overline \theta_i},
\end{equation}
see~\cite{Boza71}, and
\begin{equation}
\label{def-entropy}
D(\theta|| P)= \sum_{i,j=1}^N \theta_{ij}\log \frac{\theta_{ij}}{\overline \theta_i P_{ij}},
\end{equation}
see, e.g.,~\cite{Natarajan85}. We remark that $H(\theta)$ has the physical interpretation of the entropy rate of the Markov chain $Q$, and $D(\theta || P)$ is the relative entropy of the Markov chain $Q$ with respect to the Markov chain $P$. We refer to $H$ as the entropy function and to $D(\cdot ||P)$ as the relative entropy function.}

The following lemma asserts that $H$ is concave, and $D(\cdot||P)$ is convex (in $\theta$). We will use these results in Section~\ref{sec-Convexity}, when proving convexity of the error exponent lower bound. A (sketch of the) proof based on a certain matrix decomposition can be found in~\cite{Boza71}. We provide more direct proofs here, based on inspection of the Hessian matrix; see Appendix~\ref{app-A}.
\begin{lemma}
\label{lemma-H-functions-convexity}
\begin{enumerate}
\item \label{part-1-H-concave}
The entropy function $H: \mathbb R^{N\times N}_+\mapsto \mathbb R$ is concave.
\item \label{part-2-H-relative-convex}
The relative entropy function $D(\cdot||P):\mathbb R^{N\times N}_+\mapsto \mathbb R$ is convex.
\end{enumerate}
\end{lemma}

In Lemma~\ref{lemma-C-t-theta-bounds} that we state next we use the entropy function $H$ to approximate the cardinalities $C_{t,\theta}$ of sets $\mathcal S_{\theta}^t$, $\theta \in \Delta_t$. In particular, the result in part~\ref{part-C-t-theta} asserts that $C_{t,\theta}$ increases exponentially fast in the sequence length $t$, with the rate equal to $H(\theta)$; this result is originally proved by Whittle's formula, see~\cite{Whittle55},~\cite{Boza71}, and~\cite{Davisson81}, but we remark that it can also be proven using the asymptotic equipartition property (AEP) for Markov sources, see Chapter 3.1 in~\cite{Cover06}. It is also of interest to estimate the cardinality $C_t$ of the set $\mathcal S^t$ of all feasible sequences until time $t$. The result in part~\ref{part-C-t} of the lemma states that $C_t$ also increases exponentially in $t$, with the rate {\color{black}equal to the logarithm of the spectral radius} $\rho_0$ of $P_0$, where $P_0:=P^0$ is the zero-one sparsity matrix of the transition matrix $P$; this result easily follows from the fact that $\pi>0$, {\color{black} which then implies that the chain can start in any state, and thus} $C_t= 1^\top P_0^{t-1} 1$. For completeness, we provide proofs of both results, see Appendix~\ref{app-A}.
\begin{lemma}
\label{lemma-C-t-theta-bounds}
\begin{enumerate}
\item \label{part-C-t} For each $\epsilon>0$, there exists $t_0=t_0(\epsilon, P_0)$ such that
\begin{equation}
\label{eq-bounds-on-C-t}
\rho_0^t \leq C_t \leq \rho_0^t e^{t\epsilon},
\end{equation}
for all $t\geq t_0$, where $\rho_0>1$.
\item \label{part-C-t-theta} For each $\epsilon>0$, there exists  $t_1=t_1(\epsilon, P_0)$ such that, for each {\color{black}$t\geq t_1$} and each $\theta \in \Delta_t$, there holds
\begin{equation}
\label{eq-C-t-theta-bounds}
e^{t H(\theta)-t\epsilon} \leq C_{t,\theta}\leq e^{t H(\theta)+t\epsilon}.
\end{equation}
\end{enumerate}
\end{lemma}

\mypar{Gauss-Markov types} In our problem, we have with each sequence $s^t$ an associated sequence of Gaussian random variables $X_{s_k,k}$, $k=1,2,...,t$. From the expression for the likelihood ratio~\eqref{eq-LLR}, we see that the likelihood ratio depends on the sequence $X_{s_k,k}$ only {\color{black} through the sums of signal values $\sum_{k: s_k=i} X_{i,k}$, computed across nodes $i=1,...,N$. The intuitive interpretation of this sum is the following: if a node $i$ were aware of the presence of the walk each time when the walk visits this node, the node could sum-up the recorded signal values during each of these visits, and the result would be exactly the sum $\sum_{k: s_k=i} X_{i,k}$. Motivated by this observation, we extend the notion of Markov types to, what we call, Gauss-Markov types, by constructing a pair $(\theta,\xi)$, where $\theta$ is a Markov type associated with the sequence of states $s^t$, and $\xi$ is a vector in $\mathbb R^N$ whose $i$-th component equals the average of the signal values obtained during the visits of the random walk to node $i$.}  More generally, for any pair $(s^t, Z^t)$ where $s^t=(s_1,...,s_t)$ is a sequence of states in $V$ and $Z^t=(Z_1,...,Z_t)$ is a vector in $\mathbb R^t$, the Gauss-Markov type is defined as the pair $(\theta,\xi)$, where $\theta $ is given as in eq.~\eqref{def-Markov-type} above and, for $i=1,...,N$,
\begin{equation}\label{def-Gauss-Markov-type}
\xi_i= \frac{\sum_{k: s_k=i} Z_{k}}{K_i(s^t)}.
\end{equation}
{\color{black} Interpreting the vector $Z^t$ as the sequence of $t$ random walk signals  (without the knowledge at which locations they were recorded), vector $\xi$ is then saying what would be the average signal values at each node if the random walk took exactly path $s^t$. }

\subsection{LDP for Gauss-Markov types}
\label{subsec-LDP-for-GM}

In this subsection, we introduce the {\color{black} empirical measure of the realization of the Gauss-Markov type} and show that it satisfies the large deviations principle. We first give a formal definition of the large deviations principle~\cite{DemboZeitouni93},\cite{Hollander00}.

\begin{definition}[Large deviations principle~\cite{DemboZeitouni93}]
  It is said that a sequence of measures $\mu_t$ satisfies the large deviations principle with rate function $I$ if for every measurable set $G$ the following two conditions hold:
  \begin{enumerate}
    \item
    \begin{equation}\label{def-LDP-upper}
      \limsup_{t\rightarrow +\infty} \frac{1}{t}\log  \mu_t(G) \leq - \inf_{x\in \overline G} I(x);
    \end{equation}
    \item
    \begin{equation}\label{def-LDP-lower}
      \liminf_{t\rightarrow +\infty} \frac{1}{t} \log \mu_t(G) \geq - \inf_{x\in G^{\mathrm{o}}} I(x).
    \end{equation}
  \end{enumerate}
\end{definition}

{\color{black} For each fixed $t$, consider an experiment in which elements of $\mathcal S^t$ are drawn uniformly at random.} %As previously, we let $S^t$ be a randomly chosen element of $\mathcal S^t$.
Further, suppose that with each $s^t$ we have an associated random vector $Z_{s^t} \in \mathbb R^t$ {\color{black} of $t$ i.i.d. standard Gaussian random variables. For each $t$ we assume that vectors $Z_{s^t}$ corresponding to different sequences $s^t$ are independent, and we also assume that {\color{black} any two vectors $Z_{s^t}\in \left\{Z_{s^t}: s^t\in \mathcal S^t\right\}$ and $Z_{s^{t^\prime}}\in \left\{Z_{s^{t^\prime}}: s^{t^\prime}\in \mathcal S^{t^\prime}\right\}$, where $t\neq t^\prime$,} are also independent. Thus, for every different $t$, we have $|\mathcal S^t|$  independent $t$-dimensional standard Gaussian vectors, each corresponding to a fixed sequence $s^t$.
We denote by $(\Omega,\mathcal F,\mathbb P)$ the probability space that generates the sequence of collections of these random vectors, $\{Z_{s^t}\sim \mathcal N(0,I): s^t\in \mathcal S^t\}$, $t=1,2,...$.}  %We assume that, given $S^t=s^t$, $Z_{S^t}$ is independent of $S^t$.

On $\mathcal S^t$, we define the following mappings:
\begin{align}
\label{def-Theta-t}
[\Theta_t]_{ij} (s^t)& = \frac{K_{ij} (s^t)}{t},\mathrm{\;for\;} i,j=1,...,N\\
[\overline {\Theta_t}]_{i} (s^t)& = \frac{{K_{i}} (s^t)}{t},\mathrm{\;for\;} i=1,...,N.
\end{align}
That is, $\Theta_t$ is the Markov type of the sequence $s^t$, and $\overline{\Theta}_t$ satisfies $\overline \Theta_t= \Theta_t 1$. Note that for each $s^t$ and for each $i$ $[\overline \Theta_t]_{i}$ equals $\sum_{j=1}^N [{\Theta_t}]_{ij}$.
Also, for each $\omega\in \Omega$, for each $t$ and $s^t$, define
\begin{equation}
\label{def-mathcal-Z-t}
\mathcal Z_{t,i}^{\omega} (s^t)= \left\{
\begin{array}{ll}
\frac{ \sum_{k\in [1,t]: s_k=i} Z_{s^t,k} (\omega)}{t}, & \mathrm{if\;} i \mathrm{\;is\;s.t.\;} K_i(s^t)>0\\
0,& \mathrm{otherwise}
\end{array}\right., \mathrm{\;for\;}i=1,...,N.
\end{equation}
We can see that $(\Theta_t,\mathcal Z_t^{\omega})$ is the Gauss-Markov type of the pair
$\left(s^t,Z_{s^t}(\omega)\right)$. %Note that $\mathcal Z_{t,i}^{\omega}$ has two sources of randomness: the first originating from the randomly chosen sequence $s^t$, and the second from the random realization of the Gaussian vector $Z_{s^t}= (Z_{s^t,1},..., Z_{s^t,t})\in \mathbb R^t$.
We see that, for any given $s^t$, for each $i$ such that $ K_i(s^t)>0$, $\mathcal Z_{t,i}^{\omega} (s^t)$ is a Gaussian random variable with mean $0$ and variance equal to $ K_i(s^t)/t^2= \overline{\Theta}_{t,i}(s^t)/t$. On the other hand, if for some $i$, $ K_i(s^t)=0$, then $\mathcal Z_{t,i}^{\omega} (s^t)$ is deterministic and thus has zero variance ( also equal to $K_i(s^t)/t^2= \overline{\Theta}_{t,i}(s^t)/t$). Thus, for each given $s^t$, we can write
\begin{equation}\label{def-mathcal-Z-t-distribution}
\mathcal Z_{t,i}^{\omega} (s^t) \sim \mathcal N\left(0, \frac{\overline{\Theta}_{t,i}(s^t)}{t}\right),\;\;\mathrm{\;for\;}i=1,...,N.
\end{equation}

{\color{black} For each outcome $\omega\in \Omega$, let $Q_t^\omega: \mathcal B{\left(\mathbb R^{N^2+N}\right)} \mapsto [0,1]$ denote the empirical measure of the realization of $(\Theta_t, \mathcal Z_t^\omega)$}:
\begin{equation}
\label{def-Q-t}
Q_t^{\omega} (B):= \frac{\sum_{ s^t\in \mathcal S^t} 1_{\left\{ {\color{black}(\Theta_t(s^t), \mathcal Z_t^\omega(s^t))} \in B\right\}}} {C_t},
\end{equation}
for arbitrary $B\in \mathcal B{\left(\mathbb R^{N^2+N}\right)}$, where, we recall, $C_t= |\mathcal S^t|$.
{\color{black} Or, in words, if $s^t$ is chosen uniformly at random from $\mathcal S^t$, then $Q^\omega_t (B)$ is the probability that $\left(\Theta_t(s^t), \mathcal Z_t^\omega(s^t)\right)\in B$. }(It is easy to verify that $Q_t^{\omega}$ is a probability measure.) The next result asserts that $Q_t^\omega$ satisfies the LDP with probability one (with respect to the probability law $\mathbb P$ that generates the random families $\left\{Z_{s^t}:s^t\in \mathcal S^t\right\}$) and computes the corresponding rate function. The proof of Theorem~\ref{theorem-Q-t-satisfies-LDP} is given in Section~\ref{sec-LDP-for-Gauss-Markov-proof}.
\begin{theorem}
\label{theorem-Q-t-satisfies-LDP}
For every measurable set $G$, the sequence of measures $Q_t^{\omega}$, $t=1,2,...,$ with probability one satisfies both the LDP upper bound~\eqref{def-LDP-upper} and the LDP lower bound~\eqref{def-LDP-lower}, with the same rate function $I: \mathbb R^{N^2+N}\mapsto \mathbb R$, equal for all sets $G$. The rate function $I$ is given by
\begin{equation}
I(\theta,\xi)=\left\{
\begin{array}{ll}
\log \rho_0 - H(\theta)+ J_{\theta}(\xi), & \mathrm{if}\; \theta \in \Delta\;\mathrm{and}\;H(\theta) \geq J_{\theta}(\xi)\\
+\infty,&\mathrm{otherwise}
\end{array}
\right.,
\end{equation}
where, for any $\theta \in \mathbb {R}^{N\times N}$, function $J_{\theta}: \mathbb  R^N \mapsto \mathbb R$ is defined as $J_{\theta}(\xi):= \sum_{i: \overline \theta_i>0} \frac{1}{\overline \theta_i} \frac{\xi_i^2}{2}$, for any $\xi \in \mathbb R^N$ such that $\xi_i=0$ if $\overline{\theta}_i=0$, and $J_{\theta}(\xi)= +\infty$ otherwise, where $\overline \theta_i$ is the $i$th component of $\overline \theta$.
\end{theorem}
Theorem~\ref{theorem-Q-t-satisfies-LDP} is the core result behind the main result of the paper, namely, Theorem~\ref{theorem-main} that we state next. We dedicate Section~\ref{sec-LDP-for-Gauss-Markov-proof} to proving Theorem~\ref{theorem-Q-t-satisfies-LDP}.

\section{Lower bound on $\zeta$}
\label{sec-Lower-bound}

We are now ready to state our main result on the lower bound on $\zeta$.
\begin{theorem}
\label{theorem-main}
There holds $\underline \zeta \leq \zeta$, where $\underline \zeta$ is the optimal value of the following
optimization problem
\begin{equation}
\begin{array}[+]{lc}
{\color{black}\underset{\theta}{\mathrm{minimize}}}  &  D(\theta|| P) +
\sum_{i: \overline \theta_i >0} \overline \theta_i \frac{(\beta_i- \frac{1}{\overline \theta_i}\xi_i)^2}{2}\\
\mathrm{subject\; to} & H(\theta)\geq J_{\theta} (\xi)\\
%& \overline \theta = \theta 1\\
& \theta \in \Delta\\
& \xi \in \mathbb R^N.
\end{array}.
\label{eq-main-LB-opt}
\end{equation}
\end{theorem}
Before proving Theorem~\ref{theorem-main} in Subsection~\ref{subsec-proof-main}, we first give some interpretations of this result and provide an application example.

\subsection{Interpretations and an application example}
\label{subsec-Discussion}
\mypar{Interpretation through Gauss-Markov type} If we inspect the optimization problem~\eqref{eq-main-LB-opt} through the lenses of the Gauss-Markov type~\eqref{def-Gauss-Markov-type}, a very intuitive interpretation emerges. First, the last three constraints ensure that any candidate $(\theta,\xi)\in \mathbb R^{N\times N}\times \mathbb R^N$ is a valid Gauss-Markov type. The constraint $H(\theta)\geq J_{\theta} (\xi)$ then filters out all pairs $(\theta,\xi)$ corresponding to sequences that are, in the long run, infeasible under $H_0$; see also Case~2 in the proof of Theorem~\ref{theorem-Q-t-satisfies-LDP}, large deviations upper bound, in Subsection~\ref{subsec-LDP-UB}. The latter condition is very intuitive (as, in the long run) under the state of nature $H=H_1$ wrong decisions $H=H_0$  can only be made on the set of $H_0$ feasible types $(\theta,\xi)$. Finally, considering the objective function~\eqref{eq-main-LB-opt}, we see that it consists of two terms. The first term has the objective of choosing the sequence $s^t$ whose Markov type $\theta$ is as close as possible to the transition matrix of the random walk -- i.e., the ``true'' transition matrix. The second term has the objective of choosing the Gaussian signal (i.e., the random walk signal) $Z_{s^t}$ whose per node sample means $\xi_i/\overline \theta_i$ are as close as possible to the expected sample means $\beta_i$, $i=1,...,N$. Hence, the objective of~\eqref{eq-main-LB-opt} aims at finding the Gauss-Markov type whose probability is highest under the state of nature $H_1$. In summary, among all types that are asymptotically feasible under $H=H_0$, optimization problem~\eqref{eq-main-LB-opt} finds the one that has the highest probability (i.e., the slowest probability decay)
under $H=H_1$.

\mypar{Condition for detectability of the random walk} From~\eqref{eq-main-LB-opt}, it is easy to see that, if
\begin{equation}
\label{eq-detectability-condition}
H(P)\geq \sum_{i=1}^N \pi_i \frac{\beta_i^2}{2},
\end{equation}
then $(\theta^\star, \xi^\star)$ is an optimizer, where $\theta^\star_{ij}=\pi_i P_{ij}$, for $i,j=1,...,N$, and $\xi_i^\star = \pi_i \beta_i$, for $i=1,...,N$. The resulting optimal value of~\eqref{eq-main-LB-opt} equals zero. Hence, the lower bound $\underline \zeta$ on the error exponent equals zero, indicating that the random walk is not detectable.

\mypar{Special case $\beta_i\equiv\beta$} We next consider the special case when all sensors have the same SNR, i.e., when $\beta_i\equiv \beta$, for some $\beta \in \mathbb R$. In this case, it can be shown that problem~\eqref{eq-main-LB-opt} reduces to:
\begin{equation}
\begin{array}[+]{lc}
{\color{black}\underset{\theta,\xi}{\mathrm{minimize}}} &  D(\theta\left\|P\right.) + \frac{\left(\beta-\xi\right)^2}{2}\\
\mathrm{subject\; to} & H(\theta) \geq \frac{\xi^2}{2}\\
& \theta \in \Delta\\
& \xi \in \mathbb R.
\end{array}.
\label{eq-opt-equal-SNR}
\end{equation}
In particular, from~\eqref{eq-opt-equal-SNR}, we can easily recover the condition from~\cite{Agaskar15} for the random walk to be detectable: if the entropy of the random walk $H(P)$ is greater than the SNR,
\begin{equation}
\label{eq-detectability}
H(P)\geq \frac{\beta^2}{2},
\end{equation}
then the infimum of~\eqref{eq-opt-equal-SNR} is zero. Again, since the error exponent lower bound $\underline \zeta$ is then zero, this indicates that under condition~\eqref{eq-detectability} the random walk is not detectable. One can in fact show that optimization problem~\eqref{eq-opt-equal-SNR} yields the same value as the error exponent lower bound from~\cite{Agaskar15} (see eq.~(29) in~\cite{Agaskar15}). We omit the proof here, but we remark that the equivalence of the two optimization problems can be shown by using the method of limiting functions from~\cite{Vasek80} together with a single-letter parametrization in~\eqref{eq-solution-form-theta} of the set of candidate solutions $\theta$ (see also the proof of Lemma~\ref{lemma-H-and-R-properties} in Appendix~\ref{app-D} for how this can be achieved).

\mypar{Application example: Frequency hopping random access} We now explain how one can apply our methodology to develop a novel, frequency hopping based random access scheme. The motivating practical setting is NarrowBand IoT communications~\cite{5GMagazine16}. In this emerging standard, the aim is to design communication protocols for future IoT applications, where a large number of devices (e.g., smart electricity, gas or water meters etc.) transmit their data to a neighboring cellular base station. The intrinsic features in such communications are extremely simple communication protocols and extremely low SNR values (many such devices are deployed in basements). Due to low SNR values in this kind of systems the standard, repetition based mechanisms for random access may exhibit very low user detection rate. The situation is further exacerbated by the fact that fading can vary significantly over time over the frequency band used. Therefore, it is not possible to determine in advance what is the best frequency to send the access requesting signal. To overcome these issues, we propose a random access scheme based on Markov chain frequency hopping. We explain the setup formally. Let $P\in \mathbb R^{N\times N}$ be a given Markov chain transition matrix, and let $\mathcal F=\left\{f_1,\ldots,f_N\right\}$ denote the set of frequencies from the allocated frequency band.  Let also $\mathcal G$ denote the graph on $\mathcal F$ defined by the sparsity pattern of $P$. Then, to ask for medium access, a user transmits a sequence of signals, each at a different frequency, where the statistics of the transitions from one frequency to another are defined by the matrix $P$. The receiver then performs the optimal likelihood ratio test to detect the presence of the medium access signal. We remark that, by performing medium access based on the described Markov chain hopping, some detection power is indeed lost (in comparison with the repetition based scheme), but, on the other hand, the scheme is able to successfully combat (unknown) frequency selective fading, and this is furthermore achieved without the need for signal synchronization between sender and receiver.

\subsection{Proof of Theorem~\ref{theorem-main}}
\label{subsec-proof-main}
For notational convenience, for each $s^t\in \mathcal S^t$ denote
$\beta(s^t) = {\sqrt{\beta_{s_1}^2+\beta_{s_2}^2+...+\beta_{s_t}^2}}$ and
\begin{equation}
\label{eq-def-cal-X-beta}
\overline {\mathcal X}(s^t) = \sum_{k=1}^t \beta_{s_k} X_{s_k,k}.
\end{equation}
For each fixed $t$ define also function $\phi_t: \mathbb{R}^{C_t} \mapsto \mathbb{R}$,
\begin{equation}\label{def-phi-t}
\phi_t(x):= -\frac{1}{t} \log \left( \sum_{s^t\in \mathcal S^t} P(s^t) e^{-\frac{\beta(s^t)^2}{2} +  x_{s^t}} \right),
\end{equation}
where $x_{s^t}$ is an element of a vector $x=\left\{x_{s^t}: s^t\in \mathcal S^t\right\}\in \mathbb {R}^{C_t}$ whose index is $s^t$. Note (see eq.~\eqref{eq-LLR}) that
$\zeta$ can be expressed as the limit of expectations of functions $\phi_t$ evaluated at
$x=\left\{\overline {\mathcal X}(s^t): s^t \in \mathcal S^t\right\}$, $t=1,2,...$, i.e.,
\begin{equation}
\label{eq-zeta-compact}
\zeta = \lim_{t\rightarrow +\infty}\mathbb E_0 \left[  \phi_t
\left(\overline {\mathcal X}(s^t)\right) \right].
\end{equation}

Similarly as in~\cite{Agaskar15}, we approximate the sum in~\eqref{eq-zeta-compact} by dropping the correlations between terms  $\overline {\mathcal X}(s^t)$ that correspond to different sequences $s^t$. Specifically, for each sequence $s^t$, we replace the Gaussian vector $(X_{s_1,1},...,X_{s_t,t})$ by the Gaussian vector of the same dimension $t$, $Z_{s^t}=(Z_{s^t,1},...,Z_{s^t,t})$, where each component $Z_{s^t,k}$, $k=1,..,t$, is standard Gaussian, and where different components are mutually independent. For each $t$ we assume that vectors $Z_{s^t}$ corresponding to different sequences $s^t$ are independent, and we also assume that any two vectors $Z_{s^t}\in \left\{Z_{s^t}: s^t\in \mathcal S^t\right\}$ and $Z_{s^{t^\prime}}\in \left\{Z_{s^{t^\prime}}: s^{t^\prime}\in \mathcal S^{t^\prime}\right\}$, where $t\neq t^\prime$, are also independent. Thus, for every different $t$, we have $|\mathcal S^t|$  independent $t$-dimensional standard Gaussian vectors, each corresponding to a fixed sequence $s^t$. %We denote by $(\Omega, \mathcal F, \mathbb P)$ the probability space that generates
%$\left\{Z_{s^t}: s^t\in \mathcal S^t\right\}$, $t=1,2,...$, by $\omega$ we denote an element of $\Omega$, and by $\mathbb E$ the corresponding expectation operator.

Let $\overline {\mathcal Z}(s^t)$ denote the $Z$-variables counterpart of $\overline {\mathcal X}(s^t)$:
\begin{equation}
\label{eq-def-cal-Z-beta}
\overline {\mathcal Z}(s^t) = \sum_{k=1}^t \beta_{s_k} Z_{s^t,k}.
\end{equation}
Then the following result holds. The proof of Lemma~\ref{lemma-main-proof-step-1} is given in Appendix~\ref{app-B}.
\begin{lemma}\label{lemma-main-proof-step-1}
For every $t\geq 1$, there holds
\begin{equation}\label{eq-phi-t-X-and-Z}
\mathbb E_0\left[ \phi_t \left( \overline {\mathcal X} \right)\right] \geq
\mathbb E\left[ \phi_t \left( \overline {\mathcal Z} \right)\right].
\end{equation}
\end{lemma}
From~\eqref{eq-phi-t-X-and-Z} we immediately obtain
\begin{equation}\label{eq-step-1-implies}
\zeta \geq \limsup_{t\rightarrow +\infty} \mathbb E \left[ \phi_t \left( \overline {\mathcal Z}\right)\right].
\end{equation}

The next result implies that the upper limit in the preceding relation is in fact a limit, and moreover, it asserts that this limit equals $\underline {\zeta}$. The proof of Lemma~\ref{lemma-main-proof-step-2} is given in Appendix~\ref{app-B}.
\begin{lemma}\label{lemma-main-proof-step-2}
 There holds:
 \begin{enumerate}
   \item \label{part-UI} The sequence of random variables $\phi_t\left( \overline {\mathcal Z}\right)$ is uniformly integrable.
   \item \label{part-almost-sure} With probability one,
   \begin{equation}\label{eq-almost-sure-limit}
    \lim_{t\rightarrow +\infty} \phi_t\left( \overline {\mathcal Z}\right) = \underline{\zeta}.
   \end{equation}
 \end{enumerate}
\end{lemma}

Almost sure convergence of $\phi_t$, together with its uniform integrability, implies convergence in expectation, i.e.,
we have
\begin{equation}\label{eq-in-expectation}
\lim_{t\rightarrow +\infty} \mathbb{E} \left[\phi_t\left( \overline {\mathcal Z}\right) \right]
= \underline{\zeta}.
\end{equation}
Using the preceding equality in~\eqref{eq-step-1-implies}, the claim of Theorem~\ref{theorem-main} follows.

\section{Proof of LDP for Gauss-Markov types}
\label{sec-LDP-for-Gauss-Markov-proof}
In this section we prove Theorem~\ref{theorem-Q-t-satisfies-LDP}. The LDP upper bound~\eqref{def-LDP-upper} is proven in Subsection~\ref{subsec-LDP-UB}, and the LDP lower bound~\eqref{def-LDP-lower} is proven in Subsection~\ref{subsec-LDP-LB}.

\subsection{Proof of the LDP upper bound}
\label{subsec-LDP-UB}

\mypar{Upper bound for boxes} We first show the LDP upper bound for all \emph{boxes} in $\mathbb R^{N^2+N}$. Fix an arbitrary box $B=C\times D$, where $C$ is a box in $\mathbb R^{N\times N}$ and $D$ is a box in $\mathbb R^N$. Fix an arbitrary $t\geq 1$. Using the fact that the random matrix $\Theta_t$ is discrete and that it can only take values in  $\Delta_t$, we have
\begin{align*}
1_{\left\{ {\color{black}(\Theta_t(s^t), \mathcal Z^\omega_t(s^t))}\in B\right\}} & = 1_{\left\{ {\color{black}\Theta_t(s^t)}\in C\right\}}
1_{\left\{ {\color{black}\mathcal Z^\omega_t (s^t)}\in D\right\}} \nonumber\\
& = \sum_{\theta \in C\cap \Delta_t} 1_{\left\{ {\color{black}\Theta_t(s^t)} = \theta\right\}}
1_{\left\{ {\color{black}\mathcal Z^\omega_t (s^t)}\in D\right\}} .
\end{align*}
Thus,
\begin{align*}
Q_t^{\omega}(B) &= \frac{\sum_{\theta \in C \cap \Delta_t} \sum_{s^t\in \mathcal S^t}
1_{\left\{{\color{black}\Theta_t(s^t)}=\theta, {\color{black}\mathcal Z^\omega_t (s^t)} \in D\right\}} }  {C_t}.
\end{align*}
Consider now a fixed $\theta \in C\cap \Delta_t$ and let $K={k_{ij}}$ be the matrix of integers that verifies the fact that $\theta$ belongs to $\Delta_t$, {\color{black} that is, for any $i,j=1,...,N$, $k_{ij}= \theta_{ij} t$}; for $i=1,...,N$, denote also $k_i= \overline \theta_i$. Recall that $\mathcal S^t_{\theta}$ denotes all sequences (of length $t$) such that, for each $(i,j)$, the number of transitions from $i$ to $j$ equals $\theta_{ij}\,t=k_{ij}$. Then, we have
\begin{align*}
\sum_{s^t \in \mathcal S^t} 1_{\left\{{\color{black}\Theta_t(s^t)}=\theta, {\color{black}\mathcal Z^\omega_t (s^t)} \in D\right\}} =
\sum_{s^t \in \mathcal S^t_{\theta}} 1_{\left\{ {\color{black}\mathcal Z_t^{\omega} (s^t)} \in D\right\}}.
\end{align*}
Introducing, for each $\omega$ and $\theta \in \Delta_t$, a new probability measure
$Q_{t,\theta}^\omega : \mathcal B\left(\mathbb R^N\right) \mapsto \mathbb R$, defined by
\begin{equation}
\label{def-Q-t-theta}
Q_{t,\theta}^{\omega}(D^\prime):= \frac{\sum_{s^t\in \mathcal S^t_\theta}
1_{\left\{ {\color{black}\mathcal Z_t^{\omega} (s^t)} \in D^\prime\right\}} } {C_{t,\theta}},
\end{equation}
for any $D^\prime \in \mathcal B\left(\mathbb R^N\right)$, we obtain
\begin{equation}
\label{eq-Q-t-equals-sum-Q-t-theta}
Q_t^{\omega}(B) = \sum_{\theta\in C\cap \Delta_t} \frac{C_{t,\theta}}{C_t} Q_{t,\theta}^\omega(D).
\end{equation}
We now analyze the empirical distribution $Q_{t,\theta}^\omega$. Since random vectors $Z_{s^t}$,
$s^t \in \mathcal S^t$, are independent, we have that the indicator functions in the family
$\left\{1_{\left\{ {\color{black}\mathcal Z_t^{\omega} (s^t)} \in D\right\}}:\,s^t\in \mathcal S^t\right\}$ are independent -- hence they are independent in the subfamily $\mathcal S^t_{\theta} \subseteq \mathcal{S}^t$ as well. Further, any sequence $s^t \in \mathcal{S}^t_\theta$ has the same Markov type $\theta$. Thus, for each $s^t \in \mathcal S^t_\theta$, we have that ${K}_i(s^t)= \overline{\theta}_i t$, for each $i$. Recalling that, for $i=1,...,N$, $\mathcal Z_{t,i}^{\omega}(s^t) \sim \mathcal N(0, {K}_i(s^t)/t^2)$, it follows that, for each fixed $i$, $\left\{\mathcal Z_{t,i}^{\omega}(s^t):\, s^t \in \mathcal S^t_\theta\right\}$ is a family of i.i.d. Gaussian random variables, with mean $0$ and variance $\overline{\theta}_i/t$. Thus, $\left\{\mathcal Z_t^{\omega}(s^t): s^t \in \mathcal S^t_\theta\right\}$ is also i.i.d.; we denote by $q_{t,\theta}: \mathcal B\left(\mathbb{R}^N\right) \mapsto \mathbb{R}$ the corresponding probability measure -- i.e.,  $q_{t,\theta}$ is a probability measure induced by $\mathcal Z_{t}^{\omega}(s^t)$, where $s^t$ is an arbitrary element of $\mathcal{S}^t_{\theta}$.

Further, since for each fixed $s^t\in \mathcal S^t$ individual components $Z_{s^t,k}$ of vector $Z_{s^t}$ are independent, by the disjoint block theorem~\cite{Karr93}, we have that, for any fixed sequence $s^t$, the individual elements of random vector $\mathcal Z_t^{\omega}(s^t)$, $ \mathcal Z_{t,i}^{\omega}(s^t)$, $i=1,...,N$, are independent. Let $[q_{t,\theta}]_i: \mathcal B\left(\mathbb{R}\right) \mapsto \mathbb{R}$ denote marginal probability measures induced by $\mathcal Z_{t,i}^{\omega}(s^t)$, for $i=1,...,N$. Recall that  $D$ is a box, and suppose that $D =D_1\times...\times D_N$, for some arbitrary closed intervals $D_i$ in $\mathbb R$, $i=1,...,N$. Then, we have
\begin{align}
\label{eq-product-of-marginals}
q_{t,\theta}(D)&:= \mathbb P\left( \mathcal Z_t^{\omega} (s^t) \in D \right)\nonumber\\
&= \prod_{i=1}^{N}[q_{t,\theta}]_i(D_i).
\end{align}
From~\eqref{def-mathcal-Z-t-distribution}, it is easy to see that
\begin{equation}
\label{eq-marginals}
[q_{t,\theta}]_i(D_i) = \left\{
\begin{array}{ll}
\frac{t}{\sqrt{2 \pi k_i}} {\color{black}\int_{D_i}} e^{-t^2 \frac{\eta_i^2}{2 k_i}} d \eta_i
& \mathrm{if\;} k_i\geq 1\\
1,& \mathrm{if\;} 0\in D_i \mathrm{\;and\;}  k_i=0\\
0,& \mathrm{if\;} 0\notin D_i \mathrm{\;and\;} k_i=0
\end{array}\right.,
\end{equation}
where we remark that the last two equalities are due to the fact that, when $k_i=0$, then $\mathcal Z_{t,i}^{\omega}(s^t)$ is deterministic and equal to zero.

Going back to the family of indicator functions
$\left\{1_{\left\{ {\color{black}\mathcal Z_t^{\omega} (s^t)} \in D\right\}}: s^t \in \mathcal S^t_{\theta}\right\}$, we conclude that these are i.i.d. Bernoulli random variables, each with success probability equal to $q_{t,\theta}(D)$, and thus, the empirical measure $ Q_{t,\theta}^\omega (D)$ has the expected value equal to this quantity,
\begin{equation}\label{eq-expexted-value-Q-t-theta}
\mathbb E\left[ Q_{t,\theta}^\omega (D) \right] = q_{t,\theta} (D).
\end{equation}
The following lemma upper bounds and computes the exponential decay rate for the probability $q_{t,\theta}(D)$.
\begin{lemma}
\label{lemma-UB-q-t}
 For any $\epsilon>0$, there exists $t_2=t_2(\epsilon, D)$ such that, for each $t\geq t_2$ and $\theta \in \Delta_t$,
 \begin{equation}
 \label{eq-UB-q-t}
 q_{t,\theta}(D) \leq e^{t\epsilon} e^ {- t \inf_{\xi \in D} J_\theta (\xi)}.
 \end{equation}
  \end{lemma}
The proof of Lemma~\ref{lemma-UB-q-t} is given in Appendix~\ref{app-C}.

We now introduce
\begin{equation*}
b:=\inf_{ (\theta, \xi) \in C\times D} J_{\theta}(\xi) - H(\theta)
\end{equation*}
and proceed with the proof by separately analyzing the cases: 1) $b\leq 0$; and 2) $b> 0$.

\mypar{Case 1: $\mathbf{b\leq 0}$}
Fix an arbitrary $\epsilon>0$ and for each $t\geq 1$ define
\begin{equation*}
\mathcal A_{t} = \left\{ \omega:  Q^{\omega}_{t,\theta} (D)  \geq q_{t,\theta}(D) e^{t \epsilon},
\mathrm{\;for\;some\;}\theta \in \Delta_t \right\}.
\end{equation*}

By the union bound, followed by the Markov's inequality applied to each of the terms in the obtained sum, we have
\begin{align}
\label{eq-BC-for-event-A-t}
\mathbb P\left( \mathcal A_{t}\right) & \leq \sum_{\theta \in \Delta_t} \mathbb P\left(Q^{\omega}_{t,\theta} (D)  \geq q_{t,\theta}(D) e^{t \epsilon}\right)\nonumber\\
& \leq \sum_{\theta \in \Delta_t} \frac{\mathbb E\left[Q^{\omega}_{t,\theta} (D)\right]}{ q_{t,\theta}(D) e^{t\epsilon}}\nonumber\\
& \leq \frac{ (t+1)^{N^2}}{e^{t \epsilon}},
\end{align}
where in the last inequality we used~\eqref{eq-expexted-value-Q-t-theta} and the fact that each coordinate $\theta_{ij}$ takes values in the set $\{0,1,...,t\}$, and therefore the number of points in $\Delta_t$ is upper bounded by $(t+1)^{N^2}$.

The quantity in~\eqref{eq-BC-for-event-A-t} decays exponentially fast for any $\epsilon>0$ and, by the Borel-Cantelli lemma, we have $\mathbb P \left(\mathcal A_{t}\mathrm{, inf.\,often}\right)=0$. Thus, there exists a set $\Omega_0^\star\subseteq \Omega$, with $\mathbb P(\Omega_0^\star)=1$, such that, for each $\omega \in \Omega_0^\star$,
for any $\epsilon>0$,
\begin{equation}
\label{eq-aux-1}
Q_{t,\theta}^{\omega} (D)\leq q_{t,\theta} (D) e^{t\epsilon},
\end{equation}
for all $\theta \in \Delta_t$, for all $t\geq t_3$, where $t_3=t_3(\omega,\epsilon, D)$. Combining with the bounds on $C_t$ and $C_{t,\theta}$ from Lemma~\ref{lemma-C-t-theta-bounds}, together with the upper bound on $q_{t,\theta}$ from Lemma~\eqref{lemma-UB-q-t}, we obtain that, for any fixed $\epsilon>0$, for all $t\geq t_4=t_4(\omega, \epsilon, D, P_0):= \max\{t_0,t_1,t_2, t_3\}$,
\begin{equation}
\label{eq-combining}
\frac{C_{t,\theta}}{C_t} Q_{t,\theta}^{\omega} (D)\leq  \frac{1} {\rho_0^t }e^{t H(\theta) - t \inf_{\xi \in D} J_{\theta} (\xi)} e^{4 t\epsilon},
\end{equation}
for all $\theta\in \Delta_t$. Going back to equation~\eqref{eq-Q-t-equals-sum-Q-t-theta}, and applying~\eqref{eq-combining} for each $\theta \in C \cap\Delta_t$, we get
\begin{align}
\label{eq-Q-t-via-inf}
Q_t^{\omega}(B) &\leq (t+1)^{N^2} \max_{\theta \in C\cap \Delta_t} Q_{t,\theta}^{\omega}(D) \frac{C_{t,\theta}}{C_t}\nonumber\\
& \leq (t+1)^{N^2} e^{4 t\epsilon} \max_{\theta \in C\cap \Delta} \frac{1} {\rho_0^t }e^{t H(\theta) - t \inf_{\xi \in D} J_{\theta} (\xi)} \nonumber\\
& \leq (t+1)^{N^2} e^{4 t\epsilon} e^{ - t \inf_{\theta \in C\cap \Delta, \xi \in D} \log \rho_0 + J_{\theta}(\xi) - H(\theta)}.
\end{align}
Now, note that the following holds
\begin{align}
& \inf_{(\theta,\xi) \in C\cap \Delta \times D} \log \rho_0 + J_{\theta}(\xi) - H(\theta)  \nonumber \\
& = \min \left\{\inf_{(\theta,\xi) \in C\cap \Delta \times D: H(\theta) \geq J_{\theta} (\xi)} \log \rho_0 + J_{\theta}(\xi) - H(\theta),
\inf_{(\theta,\xi) \in C\cap \Delta \times D: H(\theta) \leq J_{\theta} (\xi)} \log \rho_0 + J_{\theta}(\xi) - H(\theta)\right\}\nonumber\\
& = \inf_{(\theta,\xi) \in C\cap \Delta \times D: H(\theta) \geq J_{\theta}} \log \rho_0 + J_{\theta}(\xi) - H(\theta)\label{eq-two-infs}\\
& =\inf_{(\theta,\xi) \in C \times D} I(\theta,\xi)\label{eq-recognize-rate-fcn},
\end{align}
where~\eqref{eq-two-infs} follows from the fact that $b\leq 0$ (note that, since function $J_{\theta} (\xi) - H(\theta)$ is lower semi-continuous, the set $\left\{(\theta,\xi) \in C\cap \Delta \times D: H(\theta) \geq J_{\theta} (\xi)\right\}$ is non-empty), and~\eqref{eq-recognize-rate-fcn} holds by the definition of the rate function $I$. Thus, combining~\eqref{eq-recognize-rate-fcn} with~\eqref{eq-Q-t-via-inf} proves the upper bound in Case 1.

\mypar{Case 2: $\mathbf{b>0}$}
We now prove the upper bound for the case when $b>0$. Define
\begin{equation*}
\mathcal B_t = \left\{ \omega: \sum_{s^t \in \mathcal S^t_{\theta}}
1_{\left\{ {\color{black}\mathcal Z_t^{\omega} (s^t)} \in D\right\}}\geq 1,\mathrm{\,for\,some\,} \theta \in C\cap \Delta_t\right\}.
\end{equation*}
Again, by the union bound and Markov's inequality, we have
\begin{align*}
\mathbb P\left(\mathcal B_t\right) & \leq \sum_{\theta\in C\cap \Delta_t} C_{t,\theta}\, q_{t,\theta}(D)\nonumber\\
& \leq \sum_{\theta\in C\cap \Delta_t} e^{2 t \epsilon} e^{t\, \left( H(\theta)-  \inf_{\xi \in D} J_{\theta} (\xi)\right)}\\
& \leq  (t+1)^{N^2} e^{- t (b - 2 \epsilon)},
\end{align*}
which holds for any fixed $\epsilon$ and all $t\geq t_5=t_5(\epsilon, P_0, D):=\max\{t_1,t_2\}$.
For sufficiently small $\epsilon$, the latter number decays exponentially fast with $t$. It follows by the Borel-Cantelli lemma, that, for all $\theta \in C\cap \Delta_t$, with probability one, $Q^\omega_{t,\theta} (D)=0$, for all $t\geq t_5$.

Going back to eq.~\eqref{eq-Q-t-equals-sum-Q-t-theta}, we have
\begin{align*}
Q_t^{\omega}(B) &\leq (t+1)^{N^2} \max_{\theta \in C\cap \Delta_t} Q_{t,\theta}^{\omega}(D) \frac{C_{t,\theta}}{C_t}\\
& = 0,
\end{align*}
for all $t\geq t_5$. On the other hand, since $b>0$, we have that, for any point $(\theta,\xi) \in C\times D$, $J_{\theta}(\xi) - H(\theta)>0$. Hence, we obtain
\begin{equation*}
\inf_{(\theta,\xi)\in C\times D} I(\theta,\xi) = +\infty.
\end{equation*}
Combining the preceding two identities proves the upper bound for any boxes in Case 2. This completes the proof of the LDP upper bound for boxes.

\mypar{Upper bound for compact sets} We next extend the LDP upper bound from boxes to all \emph{compact}  sets. Fix a compact set $G\subseteq \mathbb R^{N^2+N}$. Fix an arbitrary $\alpha>0$. For each point $(\theta,\xi)\in G$ draw a box $B$ around $(\theta,\xi)$ of size (width) $\delta=\delta(\theta,\xi)$ such that the infimum of $I$ over $B$ is at least $I(\theta,\xi)- \alpha$,
\begin{equation*}
\inf_{(\theta^\prime, \xi^\prime) \in B_{(\theta, \xi) } (\delta(\theta,\xi)) }  I(\theta^\prime,\xi^\prime)\geq I(\theta,\xi) - \alpha.
\end{equation*}
From the family of boxes $\left\{ B_{(\theta, \xi)}{\color{black}(} \delta(\theta,\xi) {\color{black})}:\, (\theta, \xi) \in G\right\}$, we extract a finite cover of $G$, $\left\{ B_l: l=1,...,M\right\}$ (note that this is feasible due to the fact that $G$ is compact), {\color{black} where we denote} $B_l=B_{{\color{black}(} \theta_l,\xi_l{\color{black})}} (\delta_l)$ and $\delta_l$ is the appropriate box size. Then, we have
\begin{align*}
Q_t^\omega(G)\leq \sum_{l=1}^M Q_t^\omega(B_l).
\end{align*}
Since $M$ is finite, the above inequality implies
\begin{align*}
\limsup_{t\rightarrow +\infty} \frac{1}{t}\log Q_t^\omega(G)& \leq \max_{1\leq l\leq M} \limsup_{t\rightarrow +\infty} \frac{1}{t}\log Q_t^\omega(B_l)\nonumber\\
& \leq \max_{1\leq l\leq M} - \inf_{(\theta,\xi) \in B_l} I(\theta,\xi)\\
& \leq - \min_{1\leq l\leq M} I(\theta_l,\xi_l)+ \alpha\\
& \leq - \inf_{(\theta,\xi) \in G} I(\theta,\xi)+ \alpha.
\end{align*}
Noting that $\alpha$ can be chosen arbitrarily small proves the LDP upper bound for compact sets.

\mypar{Upper bound for closed sets} To complete the proof of the LDP upper bound, it remains to show that the upper bound holds for all \emph{closed} sets.  We do this by showing that the sequence of measures $Q_t^\omega$ is exponentially tight with probability one. By Lemma 1.2.18 from~\cite{DemboZeitouni93}, this together with the upper bound for all compact sets yields the upper bound for all closed sets.

\begin{lemma}
\label{lemma-Q-t-exp-tight} With probability $1$, the sequence of measures $Q_t^\omega$, $t=1,2,...$, is exponentially tight.
\end{lemma}
\begin{proof}
To show that $Q_t^\omega$ is exponentially tight it suffices to show that the rate function has compact support. To this end, note that the variable $\theta$ must belong to the compact set $\Delta$, as otherwise $I=+\infty$. Second, $I$ is finite at a given point $(\theta, \xi)$ only if there holds $H(\theta) \geq \sum_{i: \overline \theta_i>0} \frac{1}{\overline \theta_i} \frac{\xi_i^2}{2}$ and if $\xi_i=0$ for any $i$ such that $\overline \theta_i=0$. Note further that the maximal value of the entropy function $H$ on the compact set $\Delta$ equals $\log N$. From the preceding conditions, we thus obtain that, in order for $I$ to be finite at some given $(\theta,\xi)$, $\xi$ must satisfy $\xi_i^2 \leq \overline \theta_i \log N$ (note that the case $\overline \theta_i=0$ is automatically accounted for). This therefore proves that $I$ has compact support, and hence proves Lemma~\ref{lemma-Q-t-exp-tight}.
\end{proof}

\subsection{Proof of the LDP lower bound}
\label{subsec-LDP-LB}
Let $U$ be an arbitrary open set in $\mathbb R^{N^2 + N}$. Denote $a= \inf_{(\theta,\xi)\in U} I(\theta,\xi)$ and note that $a$ can either be a finite number or $+\infty$. If $a=+\infty$, then the lower bound holds trivially. Thus, in the remainder of the proof we assume that $a\in \mathbb R$.

We claim that, for any $\alpha>0$, there exists a point $(\theta^\star,\xi^\star)=(\theta^\star,\xi^\star)(\alpha)\in U$ such that
\begin{align}
  & I(\theta^\star,\xi^\star)\leq a+\alpha, \label{eq-LB-inf-1}\\
  \mathrm{and\;} & H(\theta^\star)- J_{\theta^\star} (\xi^\star) >0. \label{eq-LB-inf-2}
\end{align}
To prove this claim, consider an arbitrary fixed $\alpha>0$. Then, by the definition of $a$, there must exist $(\theta^\prime,\xi^\prime)\in U$ such that~\eqref{eq-LB-inf-1} holds. Note that $H(\theta^\prime)- J_{\theta^\prime} (\xi^\prime)$ can either be greater than $0$ or equal to $0$ (if $H(\theta^\prime)- J_{\theta^\prime} (\xi^\prime)<0$, this would contradict the fact that $a$ is finite). If   $H(\theta^\prime)- J_{\theta^\prime} (\xi^\prime)>0$, the claim is proven. Hence, suppose that $H(\theta^\prime)- J_{\theta^\prime} (\xi^\prime)=0$ . Recall that $U$ is an open set. By the definition of rate function $I$, $\theta^\prime$ must be strictly positive in at least one entry, say $\theta^\prime_i>0$.
By the fact that $U$ is open, there must exist a point $(\theta^{\prime \prime}, \xi^{\prime \prime}) \in U$, where
$\theta^{\prime \prime}=\theta^\prime$, $\xi_j^{\prime \prime}=\xi_j^\prime$, for all $j\neq i$, and
$\xi_i^{\prime \prime}$ is chosen such that {\color{black} $|\xi_i^{\prime \prime}| < |\xi_i^\prime|$}. For $\xi^{\prime \prime}$ there holds $J_{\theta^{\prime \prime}} (\xi^{\prime \prime})<J_{\theta^{\prime \prime}} (\xi^\prime)$. Thus, choosing
 $(\theta^\star,\xi^\star)=(\theta^{\prime \prime}, \xi^{\prime \prime})$ proves~\eqref{eq-LB-inf-2}.

Next, for each $t\geq 1$, pick an arbitrary point $\theta_t$ from the set of closest neighbors of $\theta^\star$\footnote{ {\color{black} We remark that there might exist more than one point in $\Delta_t$ that is closest to $\theta^\star$, rather than unique projection to $\Delta_t$. To indicate this, we use the notation $\mathrm{Argmin}$ to denote the \emph{set} of projections of $\theta^\star$ to $\Delta_t$, rather than $\mathrm{argmin}$, which is more commonly used in minimization problems where the minimizer is unique.} } in the set $\Delta_t$,
\begin{equation}
\label{def-theta-t-opt}
\theta_t\in\mathrm{Argmin}_{\theta \in \Delta_t, \overline \theta_{i}=0 \mathrm{\;if\;}\overline \theta^\star_i=0} \|\theta - \theta^\star\|.
\end{equation} % Ray Dahlio and other thinkers of the area
Note that, since the set $\Delta_t$ gets denser with $t$, we have that $\theta_t\rightarrow \theta^\star$, as $t$ goes to infinity. We show that there exists a box $D\subseteq \mathbb R^N$ (independent of $t$) such that, for all $t$ sufficiently large, $\{\theta_t\} \times D \subseteq U$. Since $(\theta^\star,\xi^\star) \in U$ and $U$ is open, we can find a sufficiently small box $B=C\times D$ centered at $(\theta^\star,\xi^\star)$, where $C$ is a box in $\mathbb R^{N\times N}$ and $D$ is a box in $\mathbb R^N$,  such that $B$ entirely belongs to $U$. Since $\theta_t \rightarrow \theta^\star$, the tail of the sequence $\theta_t$ must belong to $C$. Thus, there exists $t_6 = t_6 (\theta^\star, C, D)$ such that for all $t\geq t_6$, $\{\theta_t\}\times D \subseteq C\times D \subseteq U$.

Similarly as in eq.~\eqref{eq-Q-t-equals-sum-Q-t-theta} in the proof of the upper bound, we have
\begin{equation*}
Q_t^\omega(U) \geq Q_t^\omega(\{\theta_t\}\times D) = Q_{t,\theta_t}^\omega(D) \frac{C_{t,\theta_t}}{C_t}.
\end{equation*}
We first show that, with probability one, the empirical measures $Q^\omega_{t,\theta_t}(D)$, $\theta \in \Delta_t$, approach their respective expectations $q_{t,\theta} (D)$.

\begin{lemma}
\label{lemma-BC-2}
For any $\epsilon>0$, with probability one, there exists $t_7=t_7(\omega,\epsilon, D)$  such that
\begin{equation}
\label{eq-LB-Q_t-theta}
Q^\omega_{t,\theta_t}(D) \geq q_{t,\theta_t}(D) (1-\epsilon),
\end{equation}
for all $t\geq t_7$.
\end{lemma}

\begin{proof}
Fix $\epsilon>0$ and for each $t\geq 1$ define
\begin{equation*}
\mathcal C_t=\left\{\omega: \left| \frac{Q^\omega_{t,\theta_t}(D)}{q_{t,\theta_t}(D)} -1 \right| \geq  \epsilon\right\}.
\end{equation*}
By Chebyshev's inequality~\cite{Karr93}, we have
\begin{align}
\label{eq-Chebyshev}
\mathbb P\left(\mathcal C_t\right)
& \leq \frac{\mathrm{Var \left[ 1_{\left\{\mathcal Z_t^{\omega} \in D\right\} }\right]} }{ \epsilon^2 C_{t,\theta_t} q^2_{t,\theta_t}(D)} \nonumber\\
& = \frac{q_{t,\theta_t}(D) (1-q_{t,\theta_t}(D))}{ \epsilon^2 C_{t,\theta_t} q^2_{t,\theta_t}(D)}\nonumber\\
& \leq \frac{1}{ \epsilon^2 C_{t,\theta_t} q_{t,\theta_t}(D)}.
\end{align}

We now lower bound $q_{t,\theta_t}(D)$.
\begin{lemma}
\label{lemma-LB-on-q-t-theta}
For each $\epsilon>0$ there exists $t_8=t_8(\epsilon, \theta^\star,D)$ such that, for all $t\geq t_8$,
\begin{equation}
\label{eq-LB-on-q-t-theta}
q_{t,\theta_t}(D) \geq e^{-t\epsilon} e^{-t \inf_{\eta\in D} J_{\theta_t} (\eta)}.
\end{equation}
\end{lemma}
The proof of Lemma~\ref{lemma-LB-on-q-t-theta} is given in Appendix~\ref{app-C}.

Having Lemma~\ref{lemma-LB-on-q-t-theta}, we are now ready to complete the proof of Lemma~\ref{lemma-BC-2}. Denote $\vartheta = H(\theta^\star) - J_{\theta^\star} (\xi^\star)$ and recall that, by~\eqref{eq-LB-inf-2}, $\vartheta>0$. Applying~\eqref{eq-LB-on-q-t-theta} and~\eqref{eq-C-t-theta-bounds} in~\eqref{eq-Chebyshev} for an arbitrary fixed $\epsilon$, we have that, for all $t\geq t_9=t_9(\omega, \epsilon, \theta^\star, D, P_0):=\max\{t_1,t_8\}$,
\begin{equation}
\label{eq-BC-2-almost}
\mathbb P\left(\mathcal C_t\right) \leq \frac{1}{\epsilon^2} \frac{1}{ e^{- 2 t \epsilon + H(\theta_t)
- \inf_{\eta \in D} J_{\theta_t} (\eta) }}.
\end{equation}
Since $H$ is a continuous function, and $\theta_t\rightarrow \theta^\star$, we have that $H(\theta_t) \geq H(\theta^\star) - \epsilon$, for all $t$ larger than some $t_{10}=t_{10}(\epsilon, \theta^\star)$. Second, because $\xi^\star \in D$, for any $\theta_t$, there must hold that
$\inf_{\eta \in D} J_{\theta_t} (\eta) \leq J_{\theta_t} (\xi^\star)$. Further, since
$J_{\theta}(\eta)= \sum_{i: \overline \theta_i>0} \frac{1}{\overline \theta_i} \frac{\xi_i^2}{2}$ is continuous in $\theta$ restricted to the coordinates $i$ in which $\theta^\star_i>0$, we have that, for all $t$ larger than some $t_{11}=t_{11}(\epsilon, \theta^\star)$,
$J_{\theta_t} (\xi^\star) \leq J_{\theta^\star} (\xi^\star)+\epsilon$ (note that for any $i$ such that $\theta_i^\star=0$ it must be that $\xi_i^\star=0$ -- hence $J_{\theta_t} (\xi^\star) = \sum_{i: \overline \theta_i^\star>0} \frac{1}{\overline \theta_{t,i}} \frac{{\xi_i^\star}^2}{2} $). Applying the preceding findings in~\eqref{eq-BC-2-almost} for $\epsilon = \vartheta/5>0$, we obtain
\begin{align}
\label{eq-BC-2-final}
\mathbb P\left(\mathcal C_t\right)& \leq \frac{1}{\epsilon^2} e^{4 t \epsilon - H(\theta^\star) +  J_{\theta^\star} (\xi^\star) }\\
& = \frac{1}{\epsilon^2} e^{-t \vartheta/5},
\end{align}
which holds for all $t\geq t_{12}=t_{12}(\omega, \epsilon, \theta^\star,D, P_0):=\max\{t_9, t_{10},t_{11}\}$. The claim of the lemma follows by the Borel-Cantelli lemma.
\end{proof}

We next combine the result of Lemma~\ref{lemma-BC-2} with the lower bounds on $q_{t,\theta_t}(D)$, $t=1,2,...$, given in Lemma~\ref{lemma-LB-on-q-t-theta}. To this end, fix an arbitrary $\epsilon>0$. Let $\Omega_1^\star$ denote the probability one set that verifies Lemma~\ref{lemma-BC-2}. We have that, for every $\omega \in \Omega_1^\star$, there exists $t_{13}=t_{13} (\omega, \epsilon, D, P_0)$ such that for all $t\geq t_{13}=t_{13}(\omega, \epsilon, D) = \max\{ t_{0}, t_{1}, t_{7}\}$
\begin{align*}
Q_t^\omega(U)&\geq q_{t,\theta_t}(D) (1-\epsilon) \frac{C_{t,\theta_t}}{C_t}\\
&\geq e^{-3 t\epsilon} e^{-t \inf_{\xi \in D} J_{\theta_t}(\xi) + H(\theta_t) - \log \rho_0} (1-\epsilon).
\end{align*}
Taking the logarithm and dividing by $t$, we obtain
\begin{equation}
\label{eq-LB-almost}
\frac{1}{t}\,\log Q_t^\omega(U) \geq -3 \epsilon -  \inf_{\xi \in D} J_{\theta_t}(\xi) + H(\theta_t) -
\log \rho_0+ \frac{\log (1-\epsilon)}{t}.
\end{equation}
As $t\rightarrow +\infty$, $\theta_t\rightarrow \theta^\star$ and we have that
$\inf_{\xi \in D} J_{\theta_t}(\xi) \rightarrow \inf_{\xi \in D} J_{\theta^\star}(\xi)$, and also
$H(\theta_t)\rightarrow H(\theta^\star)$. Thus, taking the limit in~\eqref{eq-LB-almost} yields
\begin{align*}
\liminf_{t\rightarrow +\infty} \frac{1}{t}\,\log Q_t^\omega(U) & \geq -3 \epsilon -
\inf_{\xi \in D} J_{\theta^\star}(\xi)+ H(\theta^\star) - \log \rho_0 \\
& \geq -3 \epsilon - I(\theta^\star,\xi^\star),
\end{align*}
where in the last inequality we used the fact that $\xi^\star \in D$. The latter bound holds for all $\epsilon>0$, and
hence taking the supremum over all $\epsilon>0$ yields
\begin{align*}
\liminf_{t\rightarrow +\infty} \frac{1}{t}\,\log Q_t^\omega(U) & \geq - I(\theta^\star,\xi^\star)\\
&\geq - \inf_{(\theta,\xi) \in U} I(\theta,\xi) - \alpha.
\end{align*}
Recalling that $\alpha$ is arbitrarily chosen, the lower bound is proven.

\section{Convexity of $\underline \zeta$}
\label{sec-Convexity}
In this section we prove that problem~\eqref{eq-main-LB-opt} can be reformulated as a convex optimization problem and hence it is easily solvable.
\begin{lemma}
\label{lemma-zeta-opt-cvx}
If $H(P)\geq \sum_{i=1}^N \pi_i\frac{\beta_i^2}{2}$, then the optimal value $\underline \zeta$ of~\eqref{eq-main-LB-opt} equals zero. Otherwise,~\eqref{eq-main-LB-opt} is equivalent to the following convex optimization problem and $\underline \zeta$ is computed as its optimal value:
\begin{equation}
\begin{array}[+]{lc}
{\color{black}\underset{\theta}{\mathrm{minimize}}} &  -\sum_{i,j=1}^N \theta_{ij} \log P_{ij} +  \sum_{i=1}^N {\overline \theta_i}  \frac{\beta_i^2}{2}
- 2 \sqrt {\sum_{i=1}^N {\overline \theta_i}  \frac{\beta_i^2}{2}} \sqrt {H(\theta)} \\
\mathrm{subject\; to} & \theta \in \Delta\\
\end{array}.
\label{opt-zeta-cvx}
\end{equation}
\end{lemma}

\begin{proof}

We start by taking out the dependence on vector $\xi$, which we do by solving the inner optimization in~\eqref{eq-main-LB-opt} over $\xi$, for a given $\theta$:
\begin{equation}
\begin{array}[+]{lc}
{\color{black}\underset{\xi}{\mathrm{minimize}}} &  \sum_{i: \overline \theta_i >0} \overline \theta_i \frac{(\beta_i- \frac{1}{\overline \theta_i}\xi_i)^2}{2}\\
\mathrm{subject\; to} & H(\theta)\geq \sum_{i: \overline \theta_i >0}\frac{1} {\overline \theta_i} \frac{\xi_i^2}{2}\\
\end{array}.
\label{opt-zeta-cvx-inner}
\end{equation}
Note that, for any given $\theta$, for each $i$, the corresponding optimal solution $\xi^\star_i(\theta)$ has the same sign as $\beta_i$. Hence, for each $i=1,...,N$, we can introduce the change of variables $x_i=\xi_i^2/2$ (optimal $\xi^\star$ is then obtained from optimal $x^\star$ by $\xi_i^\star  = \mathrm{sign}(\beta_i) \sqrt{2 x_i^\star}$, for any $i=1,...,N$). Optimization problem~\eqref{opt-zeta-cvx-inner} can now be written as
\begin{equation}
\begin{array}[+]{lc}
{\color{black}\underset{x}{\mathrm{minimize}}} &
\sum_{i: \overline \theta_i >0} {\overline \theta_i} \frac{\beta_i^2}{2} -
\sum_{i=1}^N {\overline \theta_i} |\beta_i| \sqrt{2x_i} + \sum_{i: \overline \theta_i >0}\frac{1}{\overline \theta_i} x_i\\
\mathrm{subject\; to} & H(\theta)\geq \sum_{i: \overline \theta_i >0} \frac{1} {\overline \theta_i}  x_i\\
& x\geq 0
\end{array}.
\label{opt-zeta-cvx-inner-ref}
\end{equation}%

It is easy to see that~\eqref{opt-zeta-cvx-inner-ref} is a convex optimization problem (in $x$) with linear constraints. If $\theta$ is such that $H(\theta)=0$, then $x^\star=0$ is a solution of~\eqref{opt-zeta-cvx-inner-ref} and thus the second term in~\eqref{eq-main-LB-opt} reduces to $\sum_{i: \overline \theta_i>0} \overline \theta_i\frac{\beta_i^2}{2}$, showing that the objective function in~\eqref{opt-zeta-cvx} is correctly evaluated (note that the third term in the objective in~\eqref{opt-zeta-cvx} vanishes for $H(\theta)=0$). Consider now those matrices $\theta \in \Delta$ such that $H(\theta)>0$. Then, there exists at least one feasible point $x$ in the interior of the constraint sets of~\eqref{opt-zeta-cvx-inner-ref}, and thus we can solve~\eqref{opt-zeta-cvx-inner-ref} by solving the corresponding KKT conditions~\cite{BoydsBook04}. We dualize only the first constraint and let $\lambda$ denote the Lagrange multiplier associated with this constraint. The corresponding Lagrangian function is given by
\begin{equation}
\label{eq-Lagrangian}
L(x,\lambda):=\sum_{i: \overline \theta_i >0}\frac{1}{\overline \theta_i} x_i - \sum_{i: \overline \theta_i >0} {\overline \theta_i} |\beta_i| \sqrt{2x_i}+ \lambda \left( \sum_{i: \overline \theta_i >0}\frac{1}{\overline \theta_i} x_i - H(\theta)\right),
\end{equation}
for $x\in \mathbb R^N$ and $\lambda \in \mathbb R$. Computing the partial derivatives with respect to $x$, we obtain, for each $i$ such that $\overline \theta_i>0$, $\frac{\partial}{\partial x_i} L(x,\lambda)= \frac{1}{\overline \theta_i} - \frac{{\overline \theta_i} |\beta_i|}{ \sqrt{2 x_i}} + \lambda \frac{1}{\overline \theta_i}$ (there is no constraint imposed on $x_i$ {\color{black} for $i$} such that $\overline \theta_i=0$). Thus the KKT conditions are:
\begin{equation}
\label{eq-KKT}
\left\{
\begin{array}{ll}
(1+\lambda)= \overline \theta_i^2 \frac{ |\beta_i|}{ \sqrt{2 x_i}},\;\mathrm{for}\;i\;\mathrm{s.t.}\;\overline \theta_i >0 \\
H(\theta)\geq \sum_{i: \overline \theta_i >0}\frac{1}{\overline \theta_i} x_i\\
\lambda \geq 0\\
\lambda \left(H(\theta)- \sum_{i: \overline \theta_i >0}\frac{1}{\overline \theta_i} x_i \right) =0
\end{array}
\right..
\end{equation}

Simple algebraic manipulations reveal that a solution to~\eqref{eq-KKT} is given by
\begin{equation}
\label{eq-solution-KKT-inner}
x_i^\star = \left\{
\begin{array}{ll}
\overline \theta_i^2 \frac{\beta_i^2}{2} \frac{H(\theta)}{{\color{black}R(\theta)}},
& \mathrm{if\;} H(\theta) < {\color{black} R(\theta)}\\
\overline \theta_i^2 \frac{\beta_i^2}{2}, & \mathrm{otherwise}
\end{array}
\right.,
\end{equation}
From~\eqref{eq-solution-KKT-inner} it can be seen that, for those $\theta$ such that $H(\theta) \geq {\color{black} R(\theta)}$, the optimal value of~\eqref{opt-zeta-cvx-inner-ref} equals $0$, and for $\theta$ such that $H(\theta) \leq {\color{black} R(\theta)}$, the optimal value
of~\eqref{opt-zeta-cvx-inner-ref} equals
\begin{equation}
\label{eq-inner-opt-objective-2}
\sum_{i=1}^N {\overline \theta_i} \frac{\beta_i^2}{2} - \sum_{i=1}^N {\overline \theta_i} \beta_i^2
\sqrt{\frac{H(\theta)}{\sum_{i=1}^N {\overline \theta_i} \frac{\beta_i^2}{2}}} + H(\theta)
= {\color{black} R(\theta)} - 2 \sqrt{ {\color{black} R(\theta)}}
\sqrt{H(\theta)}+H(\theta).
\end{equation}
We can therefore solve~\eqref{eq-main-LB-opt} by partitioning the candidate space $\theta \in \Delta$ into the following two regions:
1) $\mathcal R_1=\left\{\theta\in \Delta:\,H(\theta) \geq {\color{black} R(\theta)}\right\}$; and
2) $\mathcal R_2=\left\{\theta\in \Delta:\,H(\theta) \leq {\color{black} R(\theta)}\right\}$, and then finding the minimum, and the corresponding optimizer, among the optimal values of the two optimization problems with the above defined constraint sets. For region $\mathcal R_1$, we have to solve:
\begin{equation}
\begin{array}[+]{ll}
{\color{black}\underset{\theta}{\mathrm{minimize}}} &   D(\theta|| P) \\
\mbox{subject\; to} & H(\theta)\geq \sum_{i=1}^N {\overline \theta_i} \frac{\beta_i^2}{2}\\
& \theta \in \Delta
\end{array}.
\label{opt-zeta-cvx-1}
\end{equation}

For region $\mathcal R_2$, we have to solve
\begin{equation}
\begin{array}[+]{ll}
{\color{black}\underset{\theta}{\mathrm{minimize}}} &   - \sum_{i=1}^N \sum_{j=1}^N \theta_{ij} \log P_{ij} + \sum_{i=1}^N {\overline \theta_i} \frac{\beta_i^2}{2} - 2 \sqrt{ \sum_{i=1}^N {\overline \theta_i} \frac{\beta_i^2}{2}} \sqrt{H(\theta)} \\
\mbox{subject\; to} & H(\theta)\leq \sum_{i=1}^N {\overline \theta_i} \frac{\beta_i^2}{2}\\
& \theta \in \Delta
\end{array},
\label{opt-zeta-cvx-2}
\end{equation}
where we note that the objective is obtained by cancelling out the term $H(\theta)$ in the relative entropy,
$D(\theta||P) = - \sum_{i=1}^N \sum_{j=1}^N \theta_{ij} \log P_{ij}- H(\theta)$, with the one
in~\eqref{eq-inner-opt-objective-2}.

We show that, when $H(P) \geq R(P)=\sum_{i=1}^N \pi_i \frac{\beta_i^2}{2}$,
the corresponding solution is found by solving~\eqref{opt-zeta-cvx-1} and the optimal value of~\eqref{eq-main-LB-opt} equals zero; otherwise, solution of~\eqref{eq-main-LB-opt} is found by solving~\eqref{opt-zeta-cvx-2}.

Suppose that $H(P) \geq \sum_{i=1}^N \pi_i \frac{\beta_i^2}{2}$. It is easy to verify that,
$\theta^\star_{ij}= P_{ij} \pi_i$, $i,j=1,...,N$ is a solution to~\eqref{opt-zeta-cvx-1}. Since in this case $D(\theta||P)=0$, and since the optimal value of~\eqref{eq-main-LB-opt} is always non-negative, it follows that $\theta^\star$ is an optimizer of~\eqref{eq-main-LB-opt}.

Suppose now that $H(P) {\color{black}<} \sum_{i=1}^N \pi_i \frac{\beta_i^2}{2}$. Applying Lemma~\ref{lemma-H-functions-convexity}, we have that problem~\eqref{opt-zeta-cvx-1} is convex. Dualizing the first constraint only, the resulting KKT conditions are given as follows:
\begin{equation}
\label{eq-KKT-cvx-1}
\left\{\begin{array}{l}
(1+\lambda) \log\frac{\theta_{ij}}{{\overline \theta_i}} - \log P_{ij} +\lambda \frac{\beta_i^2}{2}= 0\\
H(\theta) \geq R(\theta) \\ %\label{eq-Slater}
\lambda\geq 0\\
\lambda\, \left(R(\theta) - H(\theta)\right)= 0\\
\theta \in \Delta
\end{array}\right.,
\end{equation}
%
%where we defined $R(\theta):=\sum_{i=1}^N \overline \theta_i \frac{\beta_i^2}{2}$, $\theta\in \mathbb R^{N\times N}$.
In order to apply the KKT conditions theorem, we analyze now the existence of Slater's point~\cite{BoydsBook04}. Suppose first that for all $\theta \in \Delta$ there holds $H(\theta)=R(\theta)$. As this region of points is already contained in the constraint set of optimization problem~\eqref{opt-zeta-cvx-2}, it follows that~\eqref{eq-main-LB-opt} can be solved by
solving~\eqref{opt-zeta-cvx-2}.

Suppose now that there exists a point $\theta\in \Delta$ such that $H(\theta)>R(\theta)$. If $\theta>0$, then $\theta$ is a Slater's point; if $\theta$ is not strictly positive then, by continuity of $H$ and $R$, we can find $\theta^\prime>0$ in the close neighborhood of $\theta$ that again satisfies $H(\theta^\prime)>R(\theta^\prime)$. By the KKT theorem, we conclude that, since such a point $\theta$ exists, then we can find a solution~\eqref{opt-zeta-cvx-1} by solving the KKT conditions~\eqref{eq-KKT-cvx-1}. {\color{black} Denote a solution of~\eqref{eq-KKT-cvx-1} by $(\theta^\star,\lambda^\star)$. Since we assumed that $H(P)< \sum_{i=1}^N \pi_i \frac{\beta_i^2}{2}$, the first three KKT conditions imply that $\lambda^\star>0$ (if $\lambda^\star$ were equal to zero, then from the first condition in~\eqref{eq-KKT-cvx-1} we have that, for each $i,j$, $\theta^\star_{ij}/\overline {\theta^\star}_i=P_{ij}$. But this then violates the second condition which requires that $H(\theta^\star)=H(P)\geq R(\theta^\star)=R(P)$)}. When combined with the fourth KKT condition, this yields that, at the optimal point $\theta^\star$, there holds $H(\theta^\star)= R(\theta^\star)$. Similarly as in the analysis in the preceding paragraph, this region of points is already contained in the constraint set of optimization problem~\eqref{opt-zeta-cvx-2} and hence a solution of~\eqref{eq-main-LB-opt} can be found by solving~\eqref{opt-zeta-cvx-2}.

In the following, we show in fact a stronger claim: when $H(P)< \sum_{i=1}^N \pi_i \frac{\beta_i^2}{2}$,
solution of~\eqref{eq-main-LB-opt} is found by solving the convex relaxation~\eqref{opt-zeta-cvx}
of~\eqref{opt-zeta-cvx-2}. We do this by  showing that, under the latter condition, there exists a solution $\theta^\star$ of~\eqref{opt-zeta-cvx} that satisfies
\begin{equation}\label{eq-theta-star-exists}
H(\theta^\star) \leq \sum_{i=1}^N {\overline \theta_i}^\star \frac{\beta_i^2}{2}.
\end{equation}

Note that, because~\eqref{opt-zeta-cvx} is convex, with affine, non-empty constraints, the optimizer is found as a solution to the following KKT conditions:
\begin{equation}
\label{eq-KKT-cvx-2}
\left\{\begin{array}{l}
 - \log P_{ij} + \frac{\beta_i^2}{2}\left(1 -   \frac{\sqrt{H(\theta)}}{ \sqrt{R(\theta)}}\right) +
 \frac{\sqrt{R(\theta)}}{ \sqrt{H(\theta)}} \log\frac{\theta_{ij}}{{\overline \theta_i}}  +  \mu_i - \mu_j+\nu= 0\\
\theta \in \Delta\\
\end{array}\right..
\end{equation}
where $\mu_i$ is the Lagrange multiplier corresponding to the constraint $e_i^\top \theta 1 = 1^\top \theta e_i$, for $i=1,..., N$, and $\nu$ is the Lagrange multiplier corresponding to the constraint $1^\top \theta 1 = 1$ (recall the definition of set $\Delta$ in~\eqref{def-Delta}). Denoting
\begin{equation}
\label{def-alpha}
\alpha = \frac{\sqrt {H(\theta)}}{\sqrt {R(\theta)}},
\end{equation}
we obtain from~\eqref{eq-KKT-cvx-2} that a solution of~\eqref{opt-zeta-cvx} must be of the form
\begin{equation}
\label{eq-solution-form-theta}
\frac{\theta_{ij}}{{\overline \theta_i}} = \frac{P_{ij}^\alpha e^{- \alpha (1-\alpha) \frac{\beta_i^2}{2}} v_j(\alpha) }{  v_i(\alpha) \lambda(\alpha)},
\end{equation}
where $\lambda(\alpha)$ and $v(\alpha)$ are, respectively, the Perron value and the right Perron vector of the matrix $M(\alpha)$ defined by
\begin{equation}
\label{def-M-alpha}
[M(\alpha)]_{ij}:=P_{ij}^\alpha e^{- \alpha (1-\alpha) \frac{\beta_i^2}{2}},
\end{equation}
for $i,j=1,...,N$, and $\alpha \in \mathbb R$. Hence, we have that the set of solution candidates can be parameterized by a parameter $\alpha \in \mathbb R$, and if for some $\alpha$ condition~\eqref{def-alpha} is satisfied, then the corresponding $\theta$ is a solution of~\eqref{opt-zeta-cvx}.

For $\alpha \in \mathbb R$, let $H(\alpha)$ denote the value of the entropy function $H$ for matrix $\theta$ parameterized with $\alpha$ as in~\eqref{eq-solution-form-theta}, i.e., $H(\alpha)=H(\theta(\alpha))$, $\alpha\in \mathbb R$, and define, similarly, $R(\alpha)=R(\theta(\alpha))$, $\alpha\in \mathbb R$. The following lemma is the core of the proof of the existence of $\theta^\star$ that satisfies~\eqref{eq-theta-star-exists}. The proof of this result is given in Appendix~\ref{app-D}; it is based on the limiting sequence of functions technique, see, e.g.,~\cite{Vasek80}.

\begin{lemma}
\label{lemma-H-and-R-properties}
There holds
\begin{enumerate}
\item \label{part-1-H-and-R-properties} $H(1)= H(P)$ and $R(1)= \sum_{i=1}^N \pi_i \frac{\beta_i^2}{2}$;
\item \label{part-2-H-and-R-properties} $H(0) = \log \rho_0$;
\item \label{part-3-H-and-R-properties} $H(\alpha) - \alpha^2 R(\alpha)$ is decreasing for $\alpha \in [0,1]$.
\end{enumerate}
\end{lemma}

Corollary~\ref{corollary-alpha-star} is immediately obtained from parts~\ref{part-1-H-and-R-properties},
\ref{part-2-H-and-R-properties}, and~\ref{part-3-H-and-R-properties} of Lemma~\ref{lemma-H-and-R-properties} and the fact that $\rho_0>1$ (proven in Lemma~\ref{lemma-C-t-theta-bounds}). In particular, by the assumption that $H(P)<\sum_{i=1}^N \pi_i\frac{\beta_i^2}{2}$ and the fact that $\log \rho_0>0$, we have that $\alpha \mapsto H(\alpha)- \alpha^2 R(\alpha)$ is strictly positive at $\alpha=0$, strictly negative at $\alpha=1$, and is decreasing on the interval $[0,1]$. Hence, the claim of Corollary~\ref{corollary-alpha-star} follows.

\begin{corollary}\label{corollary-alpha-star}
If $H(P)\leq \sum_{i=1}^N \pi_i \frac{\beta_i^2}{2}$, then there exists $\alpha^\star \in (0,1)$ such that $H(\alpha^\star)= {\alpha^\star}^2 R(\alpha^\star)$.
\end{corollary}

Corollary~\ref{corollary-alpha-star} asserts that, when $H(P) \leq \sum_{i=1}^N \pi_i \frac{\beta_i^2}{2}$, then there must exist $\alpha^\star \in (0,1)$ that satisfies~\eqref{def-alpha}. We then have $H(\theta(\alpha^\star))\leq R (\theta(\alpha^\star))$, which proves that problems~\eqref{opt-zeta-cvx} and~\eqref{opt-zeta-cvx-2} have the same optimal value. This finally implies that, when $H(P) \leq \sum_{i=1}^N \pi_i \frac{\beta_i^2}{2}$, then $\underline \zeta$ is computed by solving~\eqref{opt-zeta-cvx}.

We now prove that~\eqref{opt-zeta-cvx} is convex. The set $\Delta$ is defined through linear equalities and hence is convex. Also, the first two terms of the objective function are linear, hence convex. Thus, if we show that the last term of the objective is convex, we prove the claim.  The latter follows as a corollary of the following more general result, which we prove here.

\begin{lemma}
\label{lemma-cvx-general} Let $g,h: \mathbb R^d \mapsto \mathbb R$ be two non-negative, concave functions. Then, function $f: \mathbb R^d \mapsto \mathbb R$ defined by
\begin{equation}
\label{eq-def-f}
f(x):= - \sqrt{g(x)} \sqrt{h(x)},\;\;\; x\in \mathbb R^d,
\end{equation}
is convex.
\end{lemma}
The proof of Lemma~\ref{lemma-cvx-general} is given in Appendix~\ref{app-D}.

In order to apply Lemma~\ref{lemma-cvx-general} we need to verify that functions $\sum_{i=1}^N {\overline \theta_i} \frac{\beta_i^2}{2}$ and $H(\theta)$ are concave in $(\theta, {\overline \theta})$. The former is linear and thus concave, and concavity of the latter is proven in Lemma~\ref{lemma-H-functions-convexity}.
\end{proof}

{\color{black}
\section{Numerical results}
\label{sec-Num-Results}
The goal of this section is to show how the lower bound on the error exponent can be computed efficiently, and also to verify its tightness. To this end, in Subsection~\ref{subsec-FrankWolfe} we first illustrate how problem~\eqref{opt-zeta-cvx} can be solved numerically. For this purpose, we chose Frank-Wolfe, or conditional gradient method,~\cite{Jaggi13} (see the details below), but one can also use other algorithms, such as projected gradient~\cite{NocedalWright06}. In Subsection~\ref{subsec-Num-Comparison} we then illustrate tightness of the bound~\eqref{opt-zeta-cvx} by comparing the numerical solution with the true error exponent value obtained through Monte Carlo simulations.

\subsection{Frank-Wolfe optimization}
\label{subsec-FrankWolfe}

Frank-Wolfe method is an iterative projection-free, gradient based method, which at each iteration minimizes the linear approximation of the function (at the current candidate point) over the given domain. To explain the method in our setup, denote the cost function in~\eqref{lemma-zeta-opt-cvx} and its gradient, respectively by $\mathcal F:\mathbb R_{+}^{N\times N}\mapsto \mathbb R$ and $\nabla \mathcal F: \mathbb R_{+}^{N\times N}\mapsto \mathbb R^{N\times N}$.
%\begin{equation}\label{eq-mathcal-F}
%\mathcal F (\theta)= -\sum_{i,j=1}^N \theta_{ij} \log P_{ij} +  \sum_{i=1}^N {\overline \theta_i}  \frac{\beta_i^2}{2}
%- 2 \sqrt {\sum_{i=1}^N {\overline \theta_i}  \frac{\beta_i^2}{2}} \sqrt {H(\theta)}.
%\end{equation}
%and let also $\nabla \mathcal F\left( \theta\right)= $ denote the gradient matrix of $\mathcal F$ at an arbitrary $\theta\in \mathbb R_{+}^{N\times N}$.
It is easy to show that $\partial \mathcal F\left(\theta \right)/\partial \theta_{ij}= - \log P_{ij} + \beta_i^2/2 - {\sqrt{H(\theta)}}/{\sqrt{R(\theta)}} \beta_i^2/2 + {\sqrt{R(\theta)}}/{\sqrt{H(\theta)}} \log \left(\theta_{ij}/\bar \theta_i\right)$, at $\theta\in \mathbb R_{+}^{N\times N}$. The pseudocode of the Frank-Wolfe algorithm is given in Algorithm~\ref{Alg-Frank-Wolfe}.
\begin{algorithm}
 \caption{Frank-Wolfe algorithm} \label{Alg-Frank-Wolfe}
 \begin{algorithmic}[1]
 \renewcommand{\algorithmicrequire}{\textbf{Input:}}
 \renewcommand{\algorithmicensure}{\textbf{Output:}}
 \renewcommand{\algorithmicprint}{\textbf{break}}
 \REQUIRE $P$, $\beta_i$, $i=1,...,N$, K, $\epsilon$
 \ENSURE  $\theta^{\mathrm{opt}}$
 %\\ \textit{Initialisation}
  \STATE \textit{Initialisation} : Let $\theta^{(0)} \in \Delta$.
 %\\ \textit{LOOP Process}
  \FOR {$k = 1$ to $K$}
  \STATE Compute $G_k=\nabla \mathcal F\left( \theta^{(k)}\right)$
  \STATE Compute $X: = {\arg \min}_{X\in \Delta} \mathrm{trace} \left(X^\top G_k\right)$
  \STATE Perform line search
  \\$\gamma:={\arg \min}_{\gamma\in [0,1]} \mathcal F\left(\theta^{(k)} + \gamma\left( X-\theta^{(k)}\right)\right)$
  \STATE Update $\theta^{(k+1)}:= (1-\gamma)\theta^{(k)} + \gamma X $
  \IF { $\left| \mathcal F(\theta^{(k+1)}) - \mathcal F(\theta^{(k)})\right|\leq \epsilon$}
  \PRINT
  \ENDIF
  \ENDFOR
 \RETURN $\theta^{\mathrm{opt}}=\theta^{(k)}$
 \end{algorithmic}
 \end{algorithm}

As in our case the domain is linear (recall~\eqref{lemma-zeta-opt-cvx} and the definition of $\Delta$ in~\eqref{def-Delta}), each Frank-Wolfe optimization step is a linear program (LP), which we state below for completeness.
\begin{equation}
\begin{array}[+]{lc}
{\color{black}\underset{X}{\text{minimize}}} &  \mathrm{trace}\left( X^\top G_k \right) \\
\mathrm{subject\; to} & 1^\top X 1= 1\\
& e_i^\top X 1= 1^\top X e_i,\mathrm{\;for\;}i=1,...,N\\
& X_{ij}=0\mathrm{\;if\;}P_{ij}=0,\mathrm{\;for\;}i,j=1,...,N\\
& X \in\mathbb R_{+}^{N\times N}.
\end{array}.
\label{eq-WF-LP-step}
\end{equation}
To exploit the sparsity of $\theta$ (which follows from the condition in $\Delta$ that assigns to $\theta$ the same sparsity structure as in the transition matrix $P$), in our implementation we represented $\theta$, and hence $X$, as column vectors; e.g., for a chain on $N$ nodes, the number of optimization variables in~\eqref{eq-WF-LP-step} is then $3N$ (as opposed to $N^2$, if we were to optimize an $N$ by $N$ matrix).

\subsection{Numerical comparison}
\label{subsec-Num-Comparison}
\mypar{Simulation setup} We consider two different simulation setups. In the first setup, we consider a chain graph on $N=50$ nodes, with transition probabilities to the neighboring nodes equal to $1/3$, and the probability of staying at the same node hence equal also to $1/3$. In the second setup, we consider a star graph on $N=51$ nodes, again with uniform transitions at all nodes (including the nodes themselves). In both setups, we assume that there are two possible mean values: $\beta_1>0$ and $\beta_2>0$. For the chain graph, we let the two mean values alternate along the chain. Similarly, with the star graph, the two values alternate on the leafs, while the central node has mean value equal to $\beta_2$.  We set $\beta_1=1$ and vary $\beta_2$ from $0$ to $20$.

\mypar{Monte Carlo paths}  To compute the true error exponent, we use the fact the error exponent $\zeta$ equals to the limit $\kappa$ of the scaled log-likelihood ratios~\eqref{eq-asymptotic-KL-rate}, together with the fact that $L_t(\mathbfcal X^t)$ can be expressed iteratively trough the product~\eqref{eq-zeta-via-Lyap-exp}. That is, we let the product in~\eqref{eq-zeta-via-Lyap-exp} grow over iterations $t$, and compute the corresponding Monte Carlo path of $1/t \log L_t(\mathbfcal X^t)$ (for higher values of $\beta_2$, we use a renormalization of the product~\eqref{eq-zeta-via-Lyap-exp}, to prevent numerical underflow). As an estimate for the true error exponent, we used $1/T \log L_T(\mathbfcal X^t)$, for $T=100000$.

\mypar{Threshold values $\beta_2^\star$} For the chain topology, from the fact that $P$ is symmetric, we have that the stationary distribution equals $\pi=(1/N,...,1/N)^\top$. Hence, from~\eqref{eq-detectability-condition} we obtain that $\beta_2^\star= {\sqrt{ 4 H(P)- \beta_1^2}}=1.8424$, where $H(P)=\log(3)=1.0986$. For the star topology, it can be shown that the stationary distribution equals $\pi=( N/(3N-2), 2/(3N-2),...,2/(3N-2))^\top$. Hence, $\beta_2^\star= {\sqrt{ 2 \frac{3N-2}{2N-1}H(P) - \frac{N-1}{2N-1}\beta_1^2} }=2.2019$, where $H(P)=1.7870$.

\begin{figure}[thp]
	\centering
	\includegraphics[height=7cm, trim={4cm 8.2cm 4cm 8.5cm},clip]{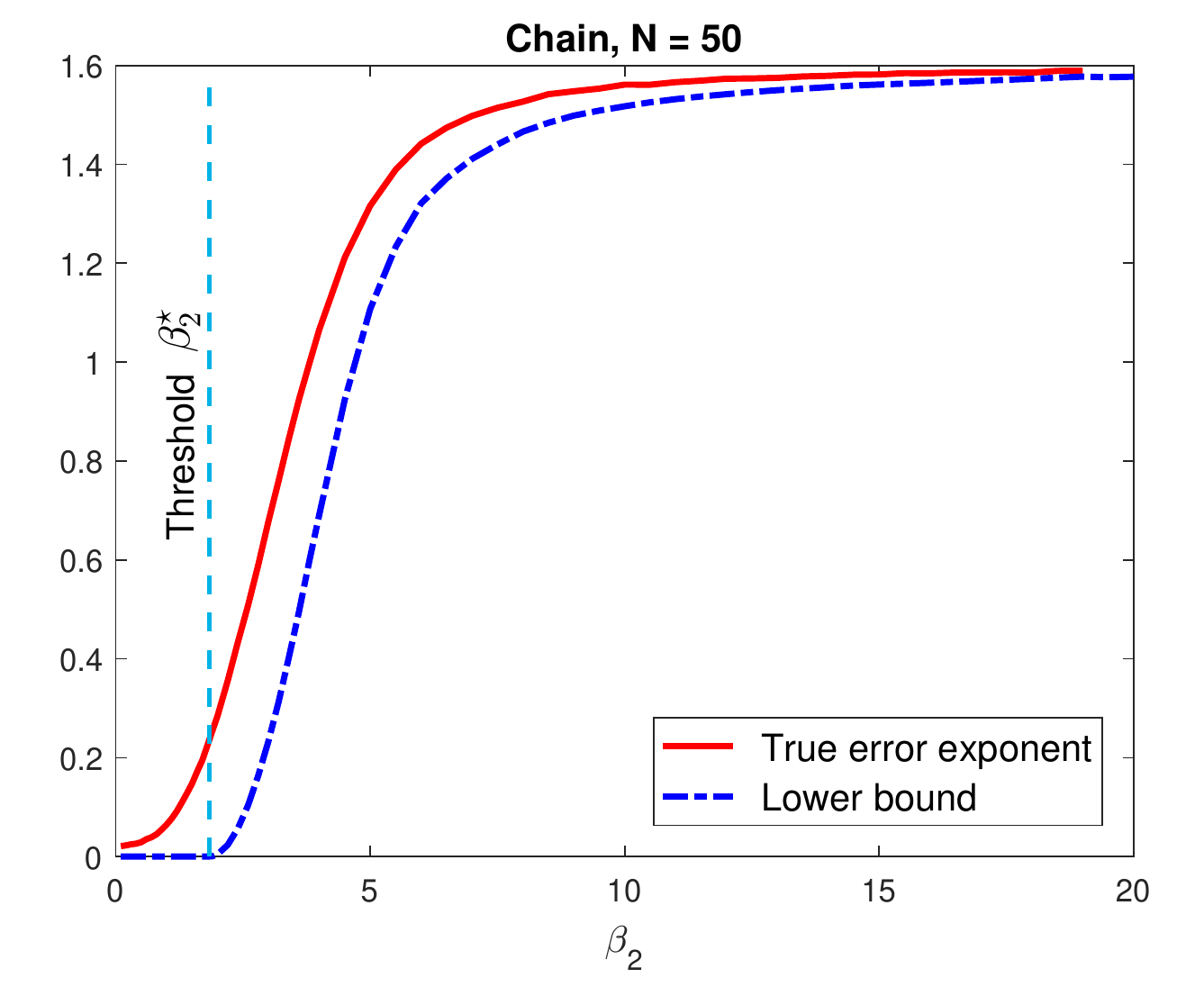}
    \includegraphics[height=7cm, trim={4cm 8.2cm 4cm 8.5cm},clip]{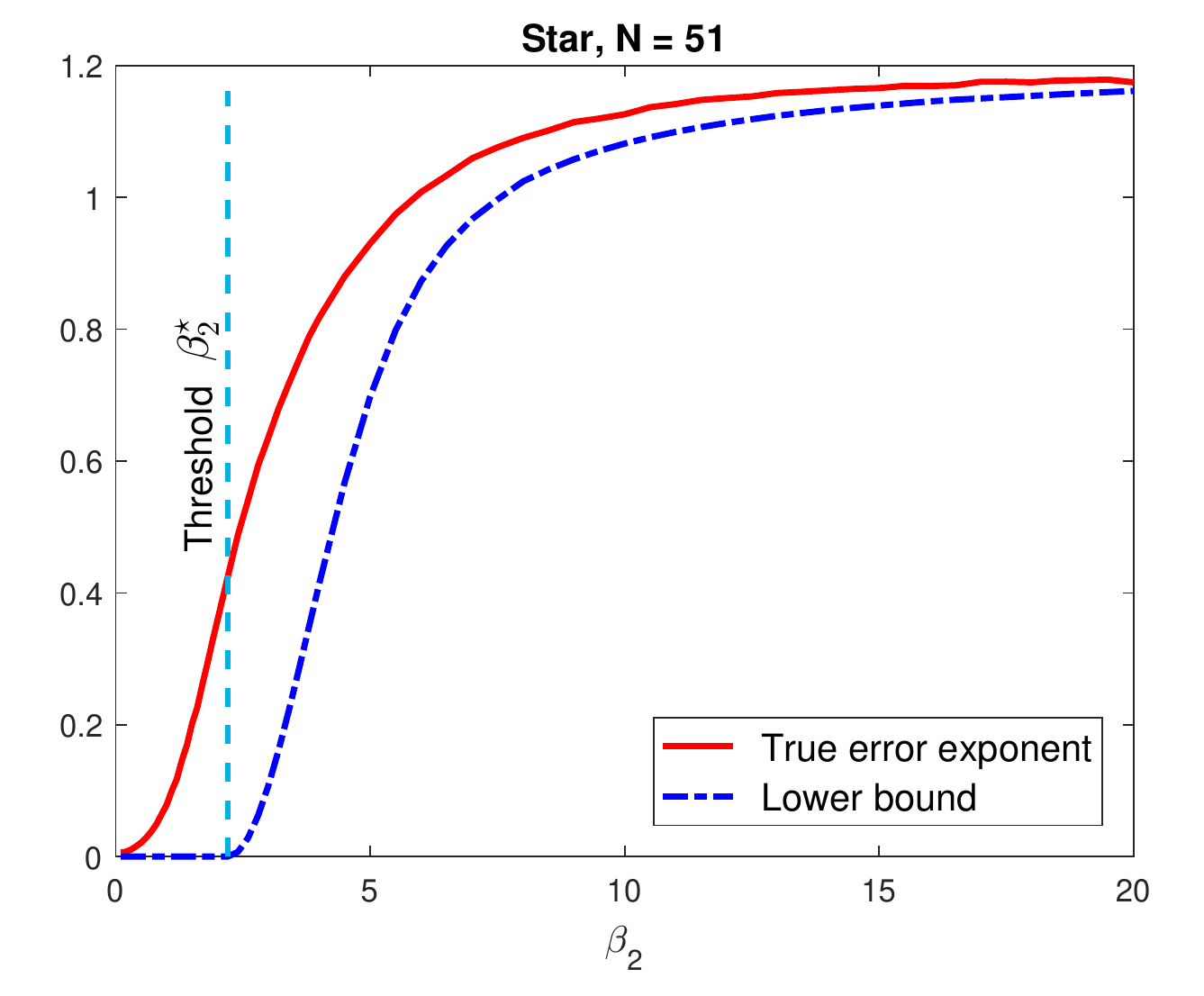}
	\caption{Simulation setup: $\mu_1=1$, $\mu_2 \in [0,20]$. Red full line plots the true error exponent; blue dotted line plots the lower bound~\eqref{opt-zeta-cvx}; light blue vertical line marks the threshold value $\beta_2^\star$ obtained from~\eqref{eq-detectability-condition}. \textbf{Top}: Chain graph on $N=50$ nodes with alternating $\beta_1$ and $\beta_2$ mean values. \textbf{Bottom}: Star graph on $N=51$ nodes with the central node's mean equal to $\beta_2$ and equal number of leaf nodes having mean values $\beta_1$ and $\beta_2$.}
	\label{Fig-Comparison}
\end{figure}

Results are shown in Figure~\ref{Fig-Comparison}. From Figure~\ref{Fig-Comparison} it can be observed that both the error exponent and the lower bound curve increase with the increase of $\beta_2$, which is expected. Also, the lower bound curve closely follows the shape of the curve corresponding to the true error exponent, while at higher values of $\beta_2$ the two curves are very close. It is interesting to note that both curves saturate (and, moreover, as it seems, at the same value). This is due to the fact that an adversarial nature can always pick the sequence of states whose Markov type is such that the transitions from low SNR to high SNR nodes occur rarely (and the low SNR nodes are more inert). To illustrate this effect, denote $Q^\star_{ij}=\theta^\star_{ij}/\overline \theta^\star_i$, where $\theta^\star$ is the solution of~\eqref{opt-zeta-cvx} and note that $Q^\star$ could be interpreted as the empirical transition matrix of the mentioned adversarial sequence. For example, with $\beta_2=4$, we obtain, for a low SNR node $i$ that the probability of staying in the same state is $Q^\star_{i\,i}=0.6493$, while the probabilities of transitioning to the high SNR neighbors are much lower, equal to $Q^\star_{i\,i-1}=Q^\star_{i\,i+1}=0.2621$. With a higher value of $\beta_2$, this difference is even more extreme, and, e.g., for $\beta_2=15$ the transition probabilities become $Q^\star_{ii}=0.9998$, while $Q^\star_{i\,i+1}=Q^\star_{i\,i-1}=0.0001$.

}

\section{Conclusion}
\label{sec-Conclusion}
We addressed the problem of detecting a random walk on a graph, when the observation interval grows large, under the assumption of heterogeneous graph nodes. We illustrate our methodology by devising a random access mechanism, which exploits Markov chain frequency hopping across different access channels to combat possibly unknown, frequency selective fading. This could be of interest for future Narrow Band IoT communications standards for which a similar random access model and similar unpredictable communication environments are envisioned.

Using the notion of Gauss-Markov type, introduced in Section~{III}, we provided in Theorem~{VIII} a lower bound on the Neyman Pearson error exponent. The bound is in the form of an optimization problem, and it has a very interesting interpretation in terms of the Gauss-Markov type: over all $H_0$-feasible Gauss-Markov types, the solution is the one which has the highest probability under $H_1$ state of nature, i.e., under the presence of the random walk. We prove convexity of the lower bound {\color{black} and illustrate how it can be efficiently evaluated using the Frank-Wolfe method for convex optimization.  Finally, we provide numerical results that the derived bound closely follows the true error exponent and is very tight in certain regions of SNRs.}

%\appendix
\section*{Appendix}
\renewcommand{\thesubsection}{\Alph{subsection}}

\subsection{ \color{white}A}
\label{app-A}
%\lipsum[2]

\mypar{Proof of Proposition~\ref{prop-UB}}

\begin{proof}
By Jensen's inequality we have
\begin{equation*}
\sum_{s^t \in \mathcal S^t} P(s^t) e^{\sum_{k=1}^t \beta_{s_k} X_{s_k,k}- \frac{\beta_{s_k}^2}{2}} \geq
e^{ \sum_{s^t  \in \mathcal S^t} P(s^t)\left\{ \sum_{k=1}^t \beta_{s_k} X_{s_k,k}- \frac{\beta_{s_k}^2}{2}  \right\}}.
\end{equation*}
Taking the logarithm and using monotonicity of the expectation, the latter inequality then implies
\begin{align}
\label{eq-Jensen-2}
\zeta &\leq - \frac{1}{t}\mathbb E_1 \left[  \sum_{s^t \in \mathcal S^t} P(s^t)\left\{ \sum_{k=1}^t \beta_{s_k} X_{s_k,k}- \frac{\beta_{s_k}^2}{2}  \right\} \right]\nonumber\\
&= - \frac{1}{t} \sum_{s^t \in \mathcal S^t} P(s^t)\left\{ \sum_{k=1}^t \beta_{s_k} \mathbb E_1[X_{s_k,k}]- \frac{\beta_{s_k}^2}{2}  \right\} \nonumber\\
&= \frac{1}{t} \sum_{s^t \in \mathcal S^t} P(s^t)\sum_{k=1}^t \frac{\beta_{s_k}^2}{2}.
\end{align}

We now observe that the sum in~\eqref{eq-Jensen-2} equals the expected value of the sum of the SNR values at nodes visited until time $t$, i.e.,
\begin{align}
\label{eq-expected-value-SNR-until-t}
\mathbb E_1\left[ \sum_{k=1}^t \frac{\beta_{S_k}^2}{2}\right] & = \sum_{s^t} P(s^t)\mathbb E_1\left[ \sum_{k=1}^t \frac{\beta_{S_k}^2}{2}\right |  S^t=s^t] \\
&= \sum_{s^t} P(s^t)\sum_{k=1}^t \frac{\beta_{s_k}^2}{2}.
\end{align}
On the other hand, the left hand-side in~\eqref{eq-expected-value-SNR-until-t} can also be written as
\begin{align}
\label{eq-expected-value-SNR-until-t-alternative}
\mathbb E_1\left[ \sum_{k=1}^t \frac{\beta_{S_k}^2}{2}\right] & =
\mathbb E_1\left[ \sum_{k=1}^t \sum_{i=1}^N 1_{\left\{ S_k=i\right\}} \frac{\beta_i^2}{2}\right]\nonumber\\
&= \sum_{k=1}^t \sum_{i=1}^N \mathbb P_1\left(S_k=i\right) \frac{\beta_i^2}{2}\nonumber\\
&= \sum_{i=1}^N \left\{\sum_{k=1}^t \mathbb P_1\left(S_k=i\right) \right\} \frac{\beta_i^2}{2}.\nonumber
\end{align}
To prove~\eqref{eq-main-UB}, it only remains to show that, for any $i=1,..,N$,
\begin{equation}
\label{eq-show-by-Cesaro}
\lim_{t\rightarrow +\infty} \frac{1}{t}\sum_{k=1}^t \mathbb P_1\left(S_k=i\right) = \pi_i.
\end{equation}
It is easy to see that $\mathbb P_1\left(S_k=i\right) = v^\top P^{k-1} e_i$. Further, since $P$ is irreducible and aperiodic, with stationary distribution $\pi$, we have that the Ces\'aro averages of $P^k$ converge to $1 \pi^\top$, i.e.,
\begin{equation*}
\lim_{t\rightarrow +\infty} \frac{1}{t}\sum_{k=1}^t P^{k-1} = 1 \pi^\top.
\end{equation*}
Finally, noting that $v^\top 1 = 1$ establishes~\eqref{eq-show-by-Cesaro} and hence proves~\eqref{eq-main-UB}.
\end{proof}

\mypar{Proof of Lemma~\ref{lemma-H-functions-convexity}}
\begin{proof}
We show that $H(\theta)$ is concave by proving that its Hessian is negative semi-definite at  $\theta\geq 0$ (and thus for all $\theta \in \Delta$). Recall that
\begin{equation*}
H(\theta)= - \sum_{i=1}^N \sum_{j=1}^N \theta_{ij} \log \frac{\theta_{ij}}{ \sum_{k=1}^N \theta_{ik}}.
\end{equation*}
It can be shown that the first and second order partial derivatives of $H(\theta)$ are given by
\begin{equation*}
\frac{\partial} {\theta_{ij}} H(\theta) = - \log \theta_{ij} + \log \left(\sum_{k=1}^N \theta_{ik}\right), \mathrm{\;for\;every\;}(i,j),
\end{equation*}
\begin{equation}
\label{eq-partial-2nd-order}
\frac{\partial^2} {\theta_{ij} \theta_{lm}} H(\theta) =
\left\{ \begin{array}{ll}
- \frac{1}{\theta_{ij}} + \frac{1} {\sum_{k=1}^N \theta_{ik}}, & \mathrm{if\;}i=l, j=m\\
\frac{1} {\sum_{k=1}^N \theta_{ik}}, & \mathrm{if\;} i= l, j\neq m\\
0, & \mathrm{otherwise}
\end{array}\right..
\end{equation}
From~\eqref{eq-partial-2nd-order} we see that the Hessian matrix $\nabla^2 H(\theta)$ takes the following block-diagonal form
\begin{equation*}
\nabla^2 H(\theta) = \mathrm{Diag}\left(B_1,...,B_N\right),
\end{equation*}
where, for each $i$, $B_i$ is an $N$ by $N$ matrix given by
\begin{equation*}
B_i = - \mathrm{Diag} \left(\frac{1}{\theta_{i1}}, \frac{1}{\theta_{i2}}, ..., \frac{1}{\theta_{iN}}\right) +
\frac{1}{\sum_{k=1}^N \theta_{ik}} 1 1^\top.
\end{equation*}
It suffices to show that each $B_i$ is, at each $\theta$, negative semi-definite. To this end, fix an arbitrary $i$, and pick an arbitrary vector $q\in \mathbb R^N$. We have,
\begin{equation}
\label{eq-quadratic-form-B_i}
q^\top B_i q = - \sum_{j=1}^N \frac{q_j^2}{\theta_{ij}} + \frac{(1^\top q)^2}{\sum_{k=1}^N \theta_{ik}}.
\end{equation}
It suffices to restrict the attention to an arbitrary $q$ such that $q\geq 0$ and $1^\top q=1$. Then,~\eqref{eq-quadratic-form-B_i} simplifies to
\begin{equation}
\label{eq-quadratic-form-B_i-simplified}
q^\top B_i q = - \sum_{j=1}^N \frac{q_j^2}{\theta_{ij}} + \frac{1}{\sum_{k=1}^N \theta_{ik}}.
\end{equation}
Exploting now Jensen's inequality with convex combination $q_1,...,q_N$, applied to the function $z\mapsto 1/z$ at points $z_j= \theta_{ij}/q_j$, yields
\begin{equation}
\sum_{j=1}^N q_j \frac{q_j}{\theta_{ij}} \geq \frac{1}{ \sum_{j=1}^N q_j \frac{\theta_{ij}}{ q_j }} = \frac{1}{\sum_{j=1}^N \theta_{jk}}.
\end{equation}
Thus, the quadratic form $q\top B_i q$ is at every $\theta$ smaller or equal than $0$, proving that $B_i$ is negative semi-definite. Since $i$ was arbitrary, it follows that $\nabla^2 H(\theta)$ is negative semi-definite at every $\theta$, which finally proves that $H$ is concave.

To prove part~\ref{part-2-H-relative-convex}, we only need to note that $H(\cdot||P)$ can be represented as a sum of three convex functions, $-H(\theta)$, $\sum_{ij}\theta_{ij} \log P_{ij}$, which is linear, hence convex, and the last is a sum of convex functions $\overline {\theta}_i\log \overline {\theta}_i$.
\end{proof}

\mypar{Proof of Lemma~\ref{lemma-C-t-theta-bounds}}

\begin{proof}
We first prove part~\ref{part-C-t}. Due to the fact that the initial distribution of the Markov chain is $\pi>0$ (and hence the chain can start from any state), it is easy to see that $C_t = 1^\top P_0^{t-1} 1$. Thus,
\begin{equation*}
\left\|P_0^{t-1}\right\|_{\infty} \leq C_t \leq N \left\|P_0^{t-1}\right\|_{\infty}.
\end{equation*}
Let $\rho$ denote the spectral radius of a square matrix and let $\rho_0$ denote the spectral radius of $P_0$. Then, by the property of spectral radius that asserts that $\rho(A)\leq \|A\|$ for any square matrix $A$ and for any matrix norm $\|\cdot\|$, we have that $\rho_0^t = \rho\left(P_0^t\right) \leq \left\|P_0^t\right\|_{\infty}$, for any $t\geq 1$. Also, by Gelfand's formula (see Theorem 8.5.1 in~\cite{MatrixAnalysis}), for any $\epsilon>0$, for all $t$ greater than some $t_0=t_0(\epsilon)$, there holds $\left\|P_0^t\right\|_{\infty} \leq \rho_0^t e^{t\epsilon}$. Summarizing, the result follows.

To prove that $\rho_0>1$, note that, since $P$ is irreducible and aperiodic, it must be that $P_0$ is also irreducible and aperiodic. Then, there exists a finite positive integer $M$ such that $P_0^M$ is strictly positive. Since each entry of $P_0^M$ is at least one, using the property that the spectral radius of an arbitrary square matrix is greater or equal than its minimal row sum~\cite{MatrixAnalysis}, we have that $\rho(P_0^M) \geq N$. It follows that $\rho(P_0) = (\rho(P_0^M))^{1/M} >1$.

To prove part~\ref{part-C-t-theta}, we use the following result, the proof of which can be found in~\cite{Hollander00}, Chapter II, Section II.2; we note that the proof of this result is based on finding the number of Euler circuits on the graph with the set of vertices $V=\{1,...,N\}$ and with the set of arcs such that the number of arcs from $i$ to $j$ equals $k_{ij}$, for each pair $(i,j)$, $i,j=1,...,N$.
\begin{lemma}
\label{lemma-via-Euler}
For each $t\geq 1$ and $\theta \in \Delta_t$ there holds
\begin{equation}
\label{eq-via-Euler-counts}
\frac{\prod_{i:\,{k_i}>0} ({ k_i}-1)!}{\prod_{i,j} k_{ij}!} \leq C_{t,\theta}
\leq N \frac{\prod_{i:\,{ k_i}>0} { k_i}!}{\prod_{i,j} k_{ij}!},
\end{equation}
where $k_{ij}=k_{ij}(\theta)$, $i,j=1,...,N$, is the matrix that verifies the fact that $\theta$ belongs to $\Delta_t$, and, for each $i$, $ k_i = \sum_{j=1}^N k_{ij}$.
\end{lemma}

To complete the proof of Lemma~\ref{lemma-C-t-theta-bounds}, we use Stirling approximation inequalities, which assert that, for all non-negative integers $k$, there holds
\begin{equation}
\label{eq-Stirling-inequalities}
\sqrt{2 \pi}\, \sqrt{k}  \left(\frac{k}{e}\right)^k   \leq   k ! \leq
e\, \sqrt{k}   \left(\frac{k}{e}\right)^k.
\end{equation}
Fixing an arbitrary $t\geq 1$ and $\theta\in \Delta_t$, we apply~\eqref{eq-Stirling-inequalities} to~\eqref{eq-via-Euler-counts}. Exploiting now the fact that $\theta_{ij}=0$ if and only if $k_{ij}=0$, and also that for $\theta_{ij}>0$, $1\leq k_{ij}\leq t$ and $1\leq k_i\leq t$, we obtain
\begin{equation*}
\frac{1}{t^N} \frac{  (2 \pi)^{\frac{N}{2}} } { e^{\frac{N^2}{2}} t^{\frac{N^2}{2}} } e^{t H(\theta)} \leq C_{t,\theta} \leq \frac{e^N t^{\frac{N}{2}}} { (2 \pi)^{\frac{N^2}{2}}}    e^{t H(\theta)}.
\end{equation*}
Finally, noting that, on both sides of the preceding inequality, the factors that premultiply $e^{t H(\theta)}$ are polynomial in $t$, and hence are dominated by any exponential function $e^{\epsilon t}$, $\epsilon>0$, for sufficiently large $t$, the claim in part~\ref{part-C-t-theta} follows.
\end{proof}

\subsection{\color{white}B}
\label{app-B}

\mypar{Proof of Lemma~\ref{lemma-main-proof-step-1}}
\begin{proof}
The main technical result behind Lemma~\ref{lemma-main-proof-step-1} is (a variant of) Slepian's lemma, see~\cite{ZeitouniNotes16}.

\begin{lemma} (Slepian's lemma~\cite{ZeitouniNotes16})
\label{lemma-Slepian}
Let the function $\phi: \mathbb R^L \mapsto \mathbb R$ satisfy
\begin{equation}
\label{slepian-first-condition}
\lim_{\|x\| \rightarrow +\infty} \phi(x) e^{- \alpha \|x\|^2}=0,\mathrm{\;for\;all\;}\alpha>0.
\end{equation}
Suppose that $\phi$ has nonnegative mixed derivatives,
\begin{equation}
\label{slepian-second-condition}
\frac{\partial^2 \phi}{\partial x_l \partial x_m} \geq 0, \mathrm{\;for\;}l\neq m.
\end{equation}
Then, for any two independent zero-mean Gaussian vectors $X$ and $Z$ taking values in $\mathbb R^L$ such that $\mathbb E_X[X_l^2]= \mathbb E_Z[Z_l^2]$ and $\mathbb E_X[X_l X_m] \geq  \mathbb E_Z[Z_l Z_m]$ there holds $\mathbb E_X[\phi(X)] \geq \mathbb E_Z [\phi(Z)]$, where $\mathbb E_X$ and $\mathbb E_Z$, respectively, denote expectation operators on probability spaces on which $X$ and $Z$ are defined.
\end{lemma}

%Fix $t\geq 1$. For each $s^t$, define
%\begin{align}
%\label{eq-X-and-Z}
%\mathcal X_{s^t}&:=\sum_{i: k_i>0} \beta_i \sum_{k\in [1,t]: s_k=i} X_{s_k,k}\\
%\mathcal Z_{s^t}&:=\sum_{i: k_i>0} \beta_i \sum_{k\in [1,t]: s_k=i} Z_{s^t,k}.
%\end{align}
%In total, we have $L = C_t$ Gaussian random variables for both $\mathcal X$ and $\mathcal Z$.
We apply Lemma~\ref{lemma-Slepian} to function $\phi_t$ in~\eqref{def-phi-t} and random variables $\overline {\mathcal X} (s^t)$ and $ \overline {\mathcal Z} (s^t)$, $s^t\in \mathcal S^t$, defined, respectively, in~\eqref{eq-def-cal-X-beta} and~\eqref{eq-def-cal-Z-beta}. We first verify the conditions of the lemma. Note that vectors $\overline {\mathcal X} $ and $ \overline  {\mathcal Z}$ are of dimension $L=C_t$.

Since $\phi_t$ grows linearly in $\|x\|$, it is easy to see that the first condition of the lemma is satisfied. Considering the second condition,
\begin{equation*}
\frac{\partial \phi_t}{\partial x_{s^t}} = -  \frac{P(s^t) e^ {- \frac{\beta(s^t)^2}{2}+ x_{s^t}}}{\sum_{s^t} P(s^t) e^ {- \frac{\beta(s^t)^2}{2}+x_{s^t}}}
\end{equation*}
and so
\begin{equation*}
\frac{\partial^2 \phi_t}{\partial x_{s^t} \partial x_{{s^t}^\prime}} = \frac{ P(s^t) e^ {- \frac{\beta(s^t)^2}{2}+ x_{s^t}} P({s^t}^\prime) e^ {- \frac{\beta_{{s^t}^\prime}^2}{2}+x_{{s^t}^\prime}}}
{\left(\sum_{s^t} P(s^t) e^ {- \frac{\beta(s^t)^2}{2}+ x_{s^t}} \right)^2},
\end{equation*}
which is non-negative for all $x\in \mathbb R^{C_t}$.

Considering further the conditions on $\overline{\mathcal X}$ and $\overline {\mathcal Z}$, we have $\mathbb E_0[\overline {\mathcal X} (s^t)^2] = \mathbb E[\overline {\mathcal Z} (s^t)^2]=\sum_{i=1}^{N} {K_i}(s^t) \frac{\beta_i^2}{2}$, for each $s^t \in \mathcal S^t$. Also, it is easy to verify that $\mathbb E_0[\overline {\mathcal X} (s^t) \overline {\mathcal X}({s^t}^\prime) ] = \sum_{k: s_k=s_k^\prime} \frac{\beta_{s_k}^2}{2} \geq 0 = \mathbb E[ \overline {\mathcal Z} (s^t) \overline {\mathcal Z} ({s^t}^\prime) ]$. Hence, all the conditions for applying Lemma~\ref{lemma-Slepian} to $\phi_t$ and random vectors $\overline {\mathcal X}$ and $\overline {\mathcal Z}$ are fulfilled. Thus, by Lemma~\ref{lemma-Slepian} the claim of Lemma~\ref{lemma-main-proof-step-1} follows.
\end{proof}

\mypar{Proof of Lemma~\ref{lemma-main-proof-step-2}}
\begin{proof}
Part~\ref{part-UI} can be proven by a simple modification of the proof in Appendix E of~\cite{Agaskar15}, hence we omit the proof here.

We prove part~\ref{part-almost-sure} by applying Varadhan's lemma, where we use the LDP for Gauss-Markov types from Theorem~\ref{theorem-Q-t-satisfies-LDP}. We use the following version of Varadhan's lemma which assumes LDP with probability one, in the sense of Theorem~\ref{theorem-Q-t-satisfies-LDP}; the proof of Lemma~\ref{lemma-Varadhan-wp1} is omitted, but we remark that it follows the line of the proof of Varadhan's lemma~\cite{DemboZeitouni93} (for deterministic sequences of measures) with each application of LDP bound (upper or lower) being used in the probability one sense, and then finally using the fact that countable intersection of probability one sets is also a probability one set.

\begin{lemma}[Varadhan's lemma~\cite{DemboZeitouni93}] \label{lemma-Varadhan-wp1}
 Suppose that for every measurable set $G$ the sequence of (random) probability measures $\mu^\omega_t$ with probability one satisfies, respectively, the LDP upper and lower bound with rate function $I$. Further, let $F: \mathbb{R}^D\mapsto \mathbb{R}$ be an arbitrary continuous function. If the tail condition~\eqref{eq-tail-condition} below holds with probability one,
 \begin{equation}\label{eq-tail-condition}
\lim_{M \rightarrow + \infty} \limsup_{t\rightarrow +\infty} \frac{1}{t}\,\log \int_{x: F(x)\geq M } e^{t F(x)} d\mu_t^\omega(x) = -\infty,
 \end{equation}
 then, with probability one,
 \begin{equation}\label{eq-Varadhan}
\lim_{t\rightarrow +\infty} \frac{1}{t} \log \sum_{x} e^{t F(x)} d\mu_t^\omega(x) = \sup_{x\in \mathbb{R}^D} F(x) - I(x).
 \end{equation}
\end{lemma}

We apply Varadhan's lemma~\ref{lemma-Varadhan-wp1} to compute the limit of the sequence $\phi_t (\overline {\mathcal Z})$, $t=1,2,...$. To this end, observe that $P(s^t)$, $\beta(s^t)$, and $\overline {\mathcal Z}(s^t)$ can be written in terms of $\Theta_t$ and $\mathcal Z_{t}$ as
\begin{align}
P(s^t)& = \frac{\pi_{s_1}}{P_{s_t, s_1}} e^{\sum_{i,j=1}^N K_{ij} (s^t) \log P_{ij} }, \label{eq-P-s-t-expressed} \\
\beta^2(s^t)& = \sum_{i=1}^N K_i(s^t) \beta_i^2, \nonumber\\
\overline {\mathcal Z}(s^t) & = t\, \sum_{i=1}^N \beta_i \,\mathcal Z_{t,i}^{\omega}(s^t).\nonumber
\end{align}
(We remark that, in~\eqref{eq-P-s-t-expressed}, it might happen that some feasible sequences $s^t$ begin and end with such $s_1=j$ and $s_t=i$ for which the transition $s_t$ to $s_1$ has zero probability, $P_{s_t,s_1}=P_{ij}=0$. Then, since $s^t$ is feasible, it must be that $K_{ij}(s^t) = 1$ (counting exactly this last, artificially added transition from $s_t$ to $s_1$), and we have that the two zero terms -- $P_{s_t,s_1}$ and $e^{ {K_{ij}(s^t)}\log P_{ij}} = P_{ij}^{K_{ij}(s^t)}$ -- will cancel out.)

Hence, we can write
\begin{align}
\phi_t &  =-\frac{1}{t} \log \left(\sum_{s^t \in \mathcal S^t} \frac{\pi_{s_1}}{P_{s_t, s_1}} e^{t \left(  \sum_{i,j=1}^N [\Theta_t]_{ij} (s^t) \log P_{ij} - \sum_{i=1}^N \overline \Theta_{t,i} (s^t)  \frac{\beta_i^2}{2}+ \sum_{i=1}^N \beta_i  \mathcal Z_{t,i}^\omega (s^t)\right)} \right)\nonumber\\
&=- \frac{1}{t} \log C_t - \frac{1}{t} \log \int_{\theta}\int_{\xi} e^{ t F (\theta,\xi)} Q_t^\omega (d \theta, d \xi),\label{eq-eta-Varadhan-form}
\end{align}
where $F: \mathbb R^{N^2+N}\mapsto \mathbb R$ is defined as
\begin{equation}
F (\theta,\xi):=\sum_{i,j=1}^N \theta_{ij} \log P_{ij} - \sum_{i=1}^N \overline{\theta}_i \frac{\beta_i^2}{2} + \sum_{i=1}^N \beta_i \xi_i,
\end{equation}
for $(\theta,\xi) \in \mathbb{R}^{N^2+N}$. By Lemma~\ref{lemma-C-t-theta-bounds}, we have that the limit of the first term equals $\log \rho_0$. To apply Varadhan's lemma to compute the limit of the second term in~\eqref{eq-eta-Varadhan-form} we first need to verify that $F$ satisfies tail condition~\eqref{eq-tail-condition}.
%
%We need to show that, with probability one, the following condition is fulfilled for function $F$:
%%
%\begin{equation}
%\label{eq-Varadhans-condition}
%\lim_{M\rightarrow + \infty} \limsup_{t\rightarrow +\infty} \frac{1}{t}\,\log \int_{(\theta,\xi): F(\theta,\xi)\geq M } e^{t F((\theta,\xi))} Q_t^\omega(d\theta,d\xi) = -\infty.
%\end{equation}
From the conditions that define the support of $I$, it is easy to conclude that the rate function has compact support. More particularly, the support $\mathcal D_I$ of the rate function will satisfy $\mathcal D_I\subseteq B_0:=\Delta \times D_0$, where $D_0 = [-\log N, +\log N]^N$  (see the proof of exponential tightness of the sequence $Q_t^\omega$, Lemma~\ref{lemma-Q-t-exp-tight}). It can be shown (similarly as in the proof of the upper bound, Case 2: $b>0$), that $Q_t^\omega (B_0)=0$, for all $t$ sufficiently large. Choosing $M_0:=\max_{(\theta,\xi) \in B_0} F(\theta,\xi)$ (note that $F$ is continuous and hence by Weierstrass theorem achieves maximum on a compact set), we have that, for each $M\geq M_0$,  with probability one, the integral in~\eqref{eq-tail-condition} equals zero for all $t$ sufficiently large. Thus, the sequence of measures $Q_t^\omega$, $t\geq 1$, with probability one, satisfies condition~\eqref{eq-tail-condition}.

By Varadhan's lemma and Lemma~\ref{lemma-C-t-theta-bounds}, part~\ref{part-C-t}, we therefore obtain that, with probability one,
\begin{align}
\label{eq-Varadhan-last-step}
 \lim_{t\rightarrow +\infty} \phi_t = \log \rho_0 - \sup_{(\theta,\xi) \in \mathbb R^{N^2+N }} F(\theta,\xi) - I(\theta,\xi).
\end{align}
It is easy to see that the optimization problem in~\eqref{eq-Varadhan-last-step} is equivalent to the one given in eq.~\eqref{eq-main-LB-opt} of Theorem~\ref{theorem-main}. This proves the claim of part~\ref{part-almost-sure} and completes the proof of the lemma.
\end{proof}

%%%%%%%%%%%%%%%%%%%%%%%%%%%%%%%%%%%%%%%%%%%%%%%%%%%%
\subsection{\color{white}C}
\label{app-C}

\mypar{Proof of Lemma~\ref{lemma-UB-q-t}}
\begin{proof}
Suppose first that there exists $i$ such that $\overline \theta_i=0$ and $0\notin D_i$. Then, by~\eqref{eq-marginals} and~\eqref{eq-product-of-marginals}, we have $q_{t,\theta}(D)=0$.  On the other hand, by the definition of $J_{\theta}$ (see Theorem~\ref{theorem-Q-t-satisfies-LDP}), we have that in this case $J_{\theta}(\xi)=+\infty$, which proves that~\eqref{eq-UB-q-t} is true.
Thus, in the remainder of the proof we assume that $0\in D_i$ for each $i$ such that $\overline \theta_i=0$.

 Fix an arbitrary $\epsilon>0$, and let $D_i=\left[\xi_i,\xi_i+\delta_i\right]$, for some $\xi_i$ and $\delta_i>0$, $i=1,...,N$. Let also $C_{\delta}= \prod_{i: \overline \theta_i>0} \delta_i$ and denote with $N_1$ the number of non-zero elements of $\overline{\theta}$. We show that, for any $\theta \in \Delta_t$,
 \begin{equation}
 \label{eq-UB-q-t-zoom-in}
 q_{t,\theta}(D) \leq \frac{t^{N_1}}{ (2 \pi)^\frac{N_1}{2} }\, C_{\delta}\, e^ {- t \inf_{\xi \in D} J_\theta (\xi)}.
 \end{equation}
It is easy to see that, for the given value of $\epsilon>0$, there exists a finite $t_2$ (that depends on $D$ and $\epsilon$) such that $\frac{t^{N_1}}{ (2 \pi)^\frac{N_1}{2} } C_{\delta} \leq e^{t \epsilon}$, for all $t\geq t_2$. Hence, if we show that the inequality~\eqref{eq-UB-q-t-zoom-in} holds, the claim of Lemma~\ref{lemma-UB-q-t} is proven.

To this end, fix $D_i$'s of the form described in the preceding paragraph and consider an arbitrarily chosen $\theta \in \Delta_t$. For each $i$, denote $k_i=\overline \theta_i t$. Note that, for each $i$ such that $\overline \theta_i>0$, $k_i$ is a positive integer by the fact that $\theta \in \Delta_t$. Consider a fixed $i$ that satisfies the preceding condition. Then, $k_i\geq 1$, and also since, for any $\eta_i \in D_i$, $e^{- \frac{t^2}{k_i} \frac{\eta_i^2}{2}} \leq e^{ - \,t \,\inf_{\eta_i\in D_i} \frac{1}{\overline \theta_i}  \frac{\eta_i^2}{2}}$, we obtain from~\eqref{eq-marginals} that
\begin{align*}
[q_{t,\theta}]_i(D_i) & \leq \frac{t} {\sqrt{2 \pi}}  e^{- t \inf_{\eta_i\in D_i} \frac{1}{\overline \theta_i} \frac{\eta_i^2}{2}} {\color{black} \int_{D_i}} d\eta_i\\
&= \frac{t} {\sqrt{2 \pi}}\, \delta_i\, e^{- t \inf_{\eta_i\in D_i} \frac{1}{\overline \theta_i} \frac{\eta_i^2}{2}} .
\end{align*}
Applying the preceding inequality in~\eqref{eq-product-of-marginals} for all $i$ such that $ k_i\geq 1$, together with~\eqref{eq-marginals} for all $i$ such that $\overline \theta_i=0$ proves~\eqref{eq-UB-q-t}.
\end{proof}

\mypar{Proof of Lemma~\ref{lemma-LB-on-q-t-theta}}
\begin{proof}
Consider first $i$ such that $\overline \theta_i^\star=0$. By~\eqref{def-theta-t-opt}, we then have $\theta_{t,i}=0$, for all $t\geq 1$. Further, by~\eqref{eq-LB-inf-2} we have that $\xi^\star =0$, and hence, the constructed interval $D_i$ contains zero. Since $\theta_{t,i}=0$, by~\eqref{eq-marginals} we therefore have $[q_{t,\theta_t}]_i(D_i)=1$.

Let $N_1$ be the number of non-zero entries in $\overline \theta^\star$. Since $\theta_t\rightarrow \theta^\star$, for any $\delta>0$, there exists $t_8^\prime = t_8^\prime (\delta,\theta^\star)$ such that for all $t\geq t_8^\prime$, there holds $\theta_{i,t} \leq (\theta^\star_i -\delta, \theta^\star_i +\delta)$. Take $\delta$ such that $\theta^\star_i -\delta>0$ for each $i$ such that $\overline \theta^\star>0$. Then, for any $t\geq t_8^\prime$ the set of non-zero elements of $\overline \theta_t$ is the same as the set of non-zero elements of $\theta^\star$, and hence, for any $\eta \in \mathbb R^N$, $J_{\theta_t}(\eta) = \sum_{i: \overline \theta^\star_i>0} \frac{1}{\overline \theta_{t,i}} \frac{\beta_i^2}{2}$. Using the fact that $D$ is a box and decomposing $q_{t,\theta_t}(D)=\prod_{i: \overline \theta_i^\star>0}[q_{t,\theta_t}]_i(D_i)$, we see that we prove the claim of the lemma if we show that, for any $i$, for any $\epsilon>0$, there exists $t_{8,i}=t_{8,i}(\epsilon, i, D_i, \theta^\star)$ such that for all $t\geq t_{8,i}$ there holds
\begin{equation}
\label{eq-LB-q-t-theta-i}
[q_{t,\theta_t}]_i(D_i) \geq e^{-t \epsilon/N_1} e^{-t \inf_{\eta_i \in D_i} \frac{1}{\theta_{i,t}} \frac{\eta_i^2}{2}}.
\end{equation}
We prove~\eqref{eq-LB-q-t-theta-i} by considering separately three cases with respect to $D_i$: 1) $\xi_i\geq 0$; 2) $\xi_i+\delta_i\leq 0$; and 3) $\xi_i\leq 0 \leq \xi_i+\delta_i$. In all three cases, we will be using the well-known bounds on the Q-function, which assert that, for an arbitrary $\alpha>0$ and $a\in \mathbb R$, there holds
\begin{equation}
\label{eq-Q-fcn-bounds}
\frac{a}{ \alpha^2 + a^2} e^{-\alpha \frac{a^2}{2}} \leq \int_{x\geq a} e^{-\alpha \frac{x^2}{2}} dx \leq \frac{1}{\alpha a} e^{-\alpha \frac{a^2}{2}}.
\end{equation}
Fix $t\geq t_8^\prime$. To simplify the notation, we drop index $t$ and denote $\theta_t$ by $\theta$, and similarly for $\overline \theta_t$. As before, we denote $k_i = t \overline \theta_i$. Note that due to the fact that $\theta\in \Delta$, we have that $k_i$ is an integer upper bounded by $t$, which together with the assumption that $\overline \theta_i>0$ implies $1\leq k_i \leq t$.

\mypar{Case 1: $\xi_i\geq 0$} We have (see~\eqref{eq-marginals}),
\begin{align}
[q_{t,\theta}]_i(D_i) & = \frac{t}{\sqrt{2 \pi k_i}} \int_{ \xi_i \leq \eta_i \leq \xi_i+\delta_i} e^{-\frac{t^2}{k_{i}} \frac{\eta_i^2}{2}}\nonumber\\
& \geq
\frac{\sqrt{ t}}{\sqrt{2 \pi}} \int_{\eta_i \geq \xi_i} e^{-\frac{t^2}{k_{i}} \frac{\eta_i^2}{2}} -
\frac{\sqrt{ t}}{\sqrt{2 \pi}} \int_{\eta_i \geq \xi_i + \delta_i} e^{-\frac{t^2}{k_{i}} \frac{\eta_i^2}{2}}\nonumber\\
& \geq
\frac{\sqrt{ t}}{\sqrt{2 \pi}} \frac{\xi_i}{ \frac{t^4}{k_{i}^2}  + \xi_i^2} e^{-\frac{t^2}{k_{i}} \frac{\xi_i^2}{2}} -
\frac{\sqrt{ t}}{\sqrt{2 \pi}}  \frac{1}{ \frac{t^2}{k_{i}} (\xi_i+\delta_i) } e^{-\frac{t^2}{k_{i}} \frac{ (\xi_i+\delta_i)^2}{2}}\nonumber\\
& \geq
\frac{\sqrt{ t}}{\sqrt{2 \pi}} \frac{\xi_i}{ t^4  + \xi_i^2} e^{-\frac{t^2}{k_{i}} \frac{\xi_i^2}{2}} -
\frac{\sqrt{ t}}{\sqrt{2 \pi}}  \frac{1}{t  (\xi_i+\delta_i)} e^{-\frac{t^2}{k_{i}} \frac{ (\xi_i+\delta_i)^2}{2}}\nonumber\\
& = \frac{\sqrt{ t}}{\sqrt{2 \pi}} \frac{\xi_i}{ t^4  + \xi_i^2}  e^ {- t \inf_{\eta_i \in D_i} \frac{1}{\overline \theta_i} \frac{\eta_i^2}{2} }
\left(1-  \frac{1}{t^{3/2}} \frac{t^4+\xi_i^2}{\xi_i(\xi_i+\delta_i)}  e^ {- t \frac{1}{\overline \theta_i} \alpha_i} \right)\label{eq-three-terms},
\end{align}
where $\alpha_i = \frac{ (\xi_i+\delta_i)^2 }{2}- \frac{\xi_i^2}{2}>0$, and where in the first inequality we used that $k_i \leq t$, and in the second we used that $1\leq k_i \leq t$.

Recall that, for all $t\geq t_8^\prime$, $\overline \theta_{i,t}\leq \overline \theta^\star_i+\delta$, thus  $\frac{1}{\overline \theta_{i,t}} \alpha_i \geq \frac{1}{\theta_i^\star + \delta} \alpha_i> 0$. This further implies that the term in the brackets in~\eqref{eq-three-terms} must converge to $1$ as $t$ increases, and hence, can be lower bounded by a positive constant between $0$ and $1$, e.g., by $1/2$ for sufficiently large $t$. On the other hand, the first term in~\eqref{eq-three-terms} decays slower than exponential, and thus the product of the first and the third term can be lower bounded by $e^{-t\epsilon/N_1}$, for $t$ sufficiently large. Thus,~\eqref{eq-LB-q-t-theta-i} holds for all $t$ greater than some $t_{8,i}=t_{8,i}(\epsilon, D_i, \theta_i^\star)$.

\mypar{Case 2: $\xi_i+\delta_i\leq 0$} Similarly as in the previous case, it can be shown here that
\begin{align*}
[q_{t,\theta}]_i(D_i) \geq \frac{\sqrt{ t}}{\sqrt{2 \pi}} \frac{  \left|\xi_i+ \delta_i\right|}{ t^4  + \left(\xi_i+ \delta_i\right)^2}  e^ {- t \inf_{\eta_i \in D_i} \frac{1}{\overline \theta_i} \frac{\eta_i^2}{2}}
\left(1-  \frac{1}{t^{3/2}} \frac{ t^4  + \left(\xi_i+ \delta_i\right)^2}
{\left|\xi_i+\delta_i\right|\left|\xi_i\right|}  e^ {- t  \frac{1}{\overline \theta_i} \beta_i} \right) ,
\end{align*}
where $\beta_i = \frac{\xi_i^2}{2}- \frac{ (\xi_i+\delta_i)^2}{2}>0$. From here, the proof is analogous as in Case 1.

\mypar{Case 3: $\xi_i \leq 0 \leq \xi_i+\delta_i$} In this case, we write
\begin{align*}
[q_{t,\theta}]_i(D_i) & = 1- \frac{t}{\sqrt{2 \pi k_i}} \int_{\eta_i \geq |\xi_i|} e^{- \frac{t^2}{k_i} \frac{\eta_i^2}{2}}d \eta_i - \frac{t}{\sqrt{2 \pi k_i}} \int_{\eta_i \geq \xi_i+ \delta_i} e^{-\frac{t^2}{k_i} \frac{\eta_i^2}{2}}d \eta_i\\
& = 1 -  \frac{t}{\sqrt{2 \pi k_i}} \frac{k_i}{t^2} \frac{1}{|\xi_i|} e^{-\frac{t^2}{k_{i}} \frac{ \xi_i^2}{2}} - \frac{t}{\sqrt{2 \pi k_i}} \frac{ k_i}{t^2} \frac{1}{\xi_i+\delta_i} e^{-\frac{t^2}{ k_{i}} \frac{ \left(\xi_i+\delta_i\right)^2}{2}}.
\end{align*}
It is easy to see that the second and the third term in the last equation go to zero as $t$ goes to infinity. Thus, for sufficiently large $t$, we have
\begin{align*}
[q_{t,\theta}]_i(D_i) & \geq  1/2.
\end{align*}
Noting that $e^{-t \inf_{\eta_i \in D_i} \frac{1}{\overline\theta_i} \frac{\eta_i^2}{2}}=1$ and that, for any $\epsilon>0$, $1/2\geq e^{-t\epsilon/N}$ for $t$ sufficiently large, proves the claim.
\end{proof}

%%%%%%%%%%%%%%%%%%%%%%%%%%%%%%%%%%%%%%%%%%%%%%%%%%%%%
\subsection{\color{white}D}
\label{app-D}
\mypar{Proof of Lemma~\ref{lemma-H-and-R-properties}}
\begin{proof}
Part~\ref{part-1-H-and-R-properties} is trivial, and the proof of part~\ref{part-2-H-and-R-properties}
can be found in~\cite{Vasek80}. We next prove part~\ref{part-3-H-and-R-properties}. We first show that the function
$\log \lambda (\alpha)$ is convex.

\begin{lemma}
\label{lemma-convexity-of-log-lambda}
 Function $\alpha \mapsto \log \lambda(\alpha)$ is convex on $\mathbb R$.
\end{lemma}

\begin{proof}
For any $\alpha\in \mathbb R$, let $f_t(\alpha)$ be defined as
\begin{equation}
\label{eq-def-f-t}
f_t(\alpha) = \frac{1}{t} \log \sum_{s^t} P^\alpha(s^t) e^{-\alpha (1-\alpha) \frac{\beta^2(s^t)}{2}}.
\end{equation}
Recalling the definition of the matrix $M(\alpha)$, it is easy to see that the sum in~\eqref{eq-def-f-t} equals
$\pi^\top M(\alpha)^{t-1} d(\alpha)$, where $d(\alpha) \in \mathbb R^N$ is a vector whose $i$-th component equals
$e^{-\alpha (1-\alpha) \frac{\beta^2_i}{2}}$. Using the fact that $\pi, d(\alpha)>0$, we obtain by Gelfand's formula (see Theorem 8.5.1 in~\cite{MatrixAnalysis}):
\begin{equation}
\label{eq-f-t-limit}
\lim_{t\rightarrow +\infty} f_t(\alpha) = \log \lambda(\alpha).
\end{equation}
Computing the first order derivative of $f_t$, we obtain
\begin{align}
\label{eq-f-t-first-order}
\frac{d}{d\alpha} f_t(\alpha) & = \frac{1}{t}  \frac{\sum_{s^t} P^\alpha(s^t) e^{-\alpha (1-\alpha) \frac{\beta^2(s^t)}{2}}
\left( \log P(s^t) +  (2\alpha -1) \frac{\beta^2(s^t)}{2} \right)}
{\sum_{s^t} P^\alpha(s^t) e^{-\alpha (1-\alpha) \frac{\beta^2(s^t)}{2}}}.
\end{align}
To ease the derivations, for any $s^t$, denote
$h_{s^t} (\alpha)= P^\alpha(s^t) e^{-\alpha (1-\alpha) \frac{\beta^2(s^t)}{2}}$ and
$g_{s^t}(\alpha)= \log P(s^t) +  (2\alpha -1) \frac{\beta^2(s^t)}{2}$, and note that
$\frac{d}{d \alpha} h_{s^t}(\alpha)= h_{s^t}(\alpha) g_{s^t}(\alpha)$. Then, it can be shown that the second order derivative of $f_t$ is given by
\begin{align}
\label{eq-f-t-second-order}
\frac{d^2}{d\alpha^2} f_t(\alpha) & = \frac{1}{t}\frac{\left(\sum_{s^t} h_{s^t}(\alpha)\right)
\left( \sum_{s^t} h_{s^t}(\alpha) g^2_{s^t}(\alpha)  +  h_{s^t}(\alpha) \beta^2(s^t) \right) -
\left(\sum_{s^t} h_{s^t}(\alpha) g_{s^t}(\alpha)\right)^2  }{ \left(\sum_{s^t} h_{s^t} (\alpha)\right)^2 }\nonumber\\
&= \frac{1}{t}\frac{\left(\sum_{s^t} h_{s^t}(\alpha)\right)
\left(\sum_{s^t} h_{s^t}(\alpha) \beta^2(s^t) \right) }{ \left(\sum_{s^t} h_{s^t} (\alpha)\right)^2 }\, +\,\nonumber\\
 & \phantom{=}\,\,\,\,\,\frac{1}{t} \frac{\left(\sum_{s^t} h_{s^t}(\alpha)\right)
 \left(\sum_{s^t} h_{s^t}(\alpha) g^2_{s^t}(\alpha)\right) -
\left(\sum_{s^t} h_{s^t}(\alpha) g_{s^t}(\alpha)\right)^2  } { \left(\sum_{s^t} h_{s^t} (\alpha)\right)^2}.\nonumber
\end{align}
The first summand in the preceding equation is positive due to the positivity of $h_{s^t} (\alpha)$ and of
$\frac{\beta^2(s^t)}{2}$. The second summand is non-negative by Cauchy-Schwartz inequality, applied to vectors
$\left\{\sqrt{h_{s^t}}\right\}_{s^t \in \mathcal S^t}$ and
$\left\{\sqrt{h_{s^t}} g_{s^t}  \right\}_{s^t \in \mathcal S^t}$. Therefore, we have that, for each $t$,
$\frac{d^2}{d\alpha^2} f_t(\alpha)\geq 0$, thus each $f_t$ is convex. Since convexity is preserved under passage to
the limit, the claim follows.
\end{proof}

For any $\alpha\in \mathbb R$, let $g_t(\alpha)$ be defined as
\begin{equation}
\label{eq-def-g-t}
g_t(\alpha) = \frac{1}{t} \log \sum_{s^t} P^\alpha(s^t) e^{-\alpha (1-\alpha) \frac{\beta^2(s^{t-1})}{2}}
v_{s_t} (\alpha).
\end{equation}
It is easy to see that the sum in~\eqref{eq-def-g-t} equals
$\pi^\top M(\alpha)^{t-1} v(\alpha) = \lambda(\alpha)^{t-1} \pi^\top v(\alpha)$. Thus
\begin{equation}\label{eq-g-t-equals}
g_t(\alpha)=\frac{1}{t}\log \pi^\top v(\alpha) +\frac{t-1}{t} \log \lambda (\alpha),
\end{equation}
which implies the following pointwise limit, on $\alpha\in \mathbb R$:
\begin{equation}
\label{eq-g-t-limit}
\lim_{t\rightarrow +\infty} g_t(\alpha) = \log \lambda(\alpha).
\end{equation}

We next compute the first order derivative of $g_t$,
\begin{align}
\label{eq-first-deriv-of-g-t}
\frac{d}{d\alpha} g_t(\alpha) & = \frac{1}{t}  \frac{\sum_{s^t} P^\alpha(s^t) e^{-\alpha (1-\alpha) \frac{\beta^2(s^{t-1})}{2}}  \left(\left( \log P (s^t) +  (2\alpha -1) \frac{\beta^2(s^{t-1})}{2} \right) v_{s_t} (\alpha) +  \frac{d}{d\alpha} v_{s_t} (\alpha)\right) }{\sum_{s^t} P^\alpha(s^t) e^{-\alpha (1-\alpha) \frac{\beta^2(s^{t-1})}{2}}v_{s_t} (\alpha)} \nonumber\\
& = \frac{1}{t}  \frac{\sum_{s^t} P^\alpha(s^t) e^{-\alpha (1-\alpha) \frac{\beta^2(s^{t-1})}{2}} v_{s_t}(\alpha) \log P(s^t) }{\lambda(\alpha)^{t-1} \pi^\top v(\alpha)} +\nonumber\\
&\frac{ (2\alpha -1)}{t}  \frac{\sum_{s^t} P^\alpha(s^t) e^{-\alpha (1-\alpha) \frac{\beta^2(s^{t-1})}{2}}  v_{s_t} (\alpha)    \frac{\beta^2(s^{t-1})}{2}  }{\lambda(\alpha)^{t-1} \pi^\top v(\alpha)}+
 \frac{1}{t}  \frac{\sum_{s^t} P^\alpha(s^t) e^{-\alpha (1-\alpha) \frac{\beta^2(s^{t-1})}{2}}
 \frac{d}{d\alpha} v_{s_t} (\alpha) }{\lambda(\alpha)^{t-1} \pi^\top v(\alpha)}.
\end{align}

We show that the sequence $\frac{d}{d\alpha} g_t(\alpha)$ is pointwise convergent on $\mathbb R$. We prove this by separately proving pointwise convergence for each of the three summation terms in~\eqref{eq-first-deriv-of-g-t}. The last term converges to zero as
$t\rightarrow +\infty$. To see this, note that:
\begin{align}
- \lambda(\alpha)^t v_{s_0}(\alpha) \delta \leq \sum_{s^t} P^\alpha(s^t) e^{-\alpha (1-\alpha) \frac{\beta^2(s^{t-1})}{2}}  \frac{d}{d\alpha} v_{s_t} (\alpha) \leq \lambda(\alpha)^t v_{s_0}(\alpha) \delta,
\end{align}
where $\delta= \max_{i=1,...,N} \frac{1}{v_i(\alpha)} \left| \frac{d}{d\alpha} v_i(\alpha)\right|$.

Recall now equation~\eqref{eq-solution-form-theta} and let $Q$ denote the respective matrix that defines the solution
$\theta$, i.e., for any $\alpha \in \mathbb R$ and any $i,j=1,...,N$, let:
\begin{equation}\label{eq-ref-matrix-Q}
[Q(\alpha)]_{ij}=\frac{P_{ij}^\alpha e^{-\alpha(1-\alpha) \frac{\beta_i^2}{2}} v_j(\alpha)}{\lambda(\alpha) v_i(\alpha)}.
\end{equation}
We observe that, for any $\alpha  \in \mathbb R$, $Q (\alpha)$ respects the sparsity pattern of matrix $P$. Thus, for any $\alpha$, $Q(\alpha)$ is irreducible and aperiodic, and therefore has a unique stationary distribution; to be consistent with~\eqref{eq-solution-form-theta}, we denote the stationary distribution of $Q(\alpha)$ by $\overline \theta=\overline \theta (\alpha)$, for any $\alpha\in\mathbb R$.

For any initial state $s_1$, denote $P\left(s^t|s_1\right)= P_{s_1s_2}\cdot \ldots \cdot P_{s_{t-1}s_t}$, and similarly for $Q\left(s^t|s_1\right)$. It is easy to verify that, for any $s^t$, $t\geq 2$, the following relation holds between $P\left(s^t|s_1\right)$ and $Q\left(s^t|s_1\right)$,
\begin{equation}
\label{eq-Q-s-t}
\frac{P^\alpha(s^t|s_1) e^{-\alpha (1-\alpha) \frac{\beta^2(s^{t-1})}{2}}  v_{s_t} (\alpha)}
{ \lambda(\alpha)^{t-1} v_{s_1}(\alpha) } = Q\left(s^t|s_1\right).
\end{equation}

Considering now the second term in~\eqref{eq-first-deriv-of-g-t}, we have for any fixed $s^t$,
\begin{align}\label{eq-second-term}
\frac{P^\alpha(s^t) e^{-\alpha (1-\alpha) \frac{\beta^2(s^{t-1})}{2}}  v_{s_t} (\alpha)\frac{\beta^2(s^{t-1})}{2}}
{\lambda(\alpha)^{t-1} \pi^\top v(\alpha)} & =  \frac{v_{s_1} (\alpha)}
{\pi^\top v(\alpha)}\frac{\pi_{s_1}  P^\alpha(s^t|s_1) e^{-\alpha (1-\alpha) \frac{\beta^2(s^{t-1})}{2}}
v_{s_t} (\alpha) \frac{\beta^2(s^{t-1})}{2} }{ \lambda(\alpha)^{t-1} v_{s_1}(\alpha) } \nonumber \\
&= \frac{\pi_{s_1}   v_{s_1} (\alpha)}{\pi^\top v(\alpha)}  Q\left(s^t|s_1\right)  \frac{\beta^2(s^{t-1})}{2}.
\end{align}
We show that the following limit holds, for any initial state $s_1$:
\begin{equation}\label{eq-limit-sum-Q-beta}
\lim_{t\rightarrow +\infty} \frac{1}{t} \sum_{s^t\setminus s_1} Q\left(s^t|s_1\right)  \frac{\beta^2(s^{t-1})}{2}
= \sum_{i=1}^{N} \overline{\theta}_i \frac{\beta_i^2}{2}.
\end{equation}
By simple algebraic manipulations, we obtain
\begin{align}
\sum_{s^t\setminus s_1} Q\left(s^t|s_1\right)  \frac{\beta^2(s^{t-1})}{2} & =
\sum_{s^t\setminus s_1} Q\left(s^t|s_1\right) \sum_{k=2}^{t} \frac{\beta_{s_{k-1}}^2}{2}\nonumber \\
%& = \sum_{k=2}^{t} \sum_{s^{k-1} \setminus s_1}  \sum_{s^t\setminus s^{k-1}}
%Q(s^{k-1}|s_1) Q(s^t|s_{k-1})  \frac{\beta_{s_{k-1}}^2}{2} \nonumber \\
& = \sum_{k=2}^{t} \sum_{s^{k-1}\setminus s_1} \frac{\beta_{s_{k-1}}^2} {2} Q\left(s^{k-1}|s_1\right)
\sum_{s^t\setminus s^{k-1}} Q\left(s^t|s_{k-1}\right)  \nonumber\\
& = \sum_{k=2}^{t} \sum_{s^{k-1}\setminus s_1} \frac{\beta_{s_{k-1}}^2} {2} Q\left(s^{k-1}|s_1\right) \label{eq-sums-up-to-1}\\
& = \sum_{k=2}^{t} \sum_{s_{k-1}} \frac{\beta_{s_{k-1}}^2} {2} Q\left(s_{k-1}|s_1\right), \label{eq-sums-up-to-1-second}
\end{align}
where in~\eqref{eq-sums-up-to-1} we use that, for any $s_{k-1}$, $\sum_{s^t\setminus s^{k-1}} Q\left(s^t|s_{k-1}\right)=1$, and in~\eqref{eq-sums-up-to-1-second} we use that, for any fixed $s_1$ and $s_{k-1}$, $\sum_{s^{k-2}\setminus s_1} Q\left(s^{k-1}|s_1\right)=Q\left(s_{k-1}|s_1\right)$. Since $Q$ is stochastic, irreducible and aperiodic, with right Perron vector $\overline \theta$, we have $Q^t \longrightarrow  1 \overline \theta^\top$, as $t \rightarrow +\infty$. Noting that $Q(s_{k-1}|s_1)= e_{s_1}^\top Q^{k-1} e_{s_{k-1}}$, for any $s_1$ and $s_{k-1}$, the limit~\eqref{eq-limit-sum-Q-beta} follows.

Going back to~\eqref{eq-first-deriv-of-g-t}, and combining~\eqref{eq-second-term} and~\eqref{eq-limit-sum-Q-beta}, we have that the limit of the second term in~\eqref{eq-first-deriv-of-g-t}, as $t\rightarrow +\infty$, equals
\begin{align}\label{eq-second-term-final}
&\lim_{t\rightarrow +\infty} \frac{ (2\alpha -1)}{t}  \frac{\sum_{s^t} P^\alpha(s^t)
e^{-\alpha (1-\alpha) \frac{\beta^2(s^{t-1})}{2}}  v_{s_t} (\alpha)    \frac{\beta^2(s^{t-1})}{2}  }
{\lambda(\alpha)^{t-1} \pi^\top v(\alpha)}\nonumber\\
& = \lim_{t\rightarrow +\infty} (2\alpha -1) \sum_{s_1}  \frac{\pi_{s_1}   v_{s_1} (\alpha)}{\pi^\top v(\alpha)} \frac{1}{t} \sum_{s^t\setminus s_1} Q\left(s^t|s_1\right)  \frac{\beta^2(s^{t-1})}{2}  \nonumber\\
 & = (2\alpha -1) \sum_{j=1}^N  \frac{\pi_{j}   v_{j} (\alpha)}{\pi^\top v(\alpha)}
  \sum_{i=1}^{N} \overline{\theta}_i \frac{\beta_i^2}{2} \nonumber \\
 & = (2\alpha -1) \sum_{i=1}^{N} \overline{\theta}_i \frac{\beta_i^2}{2},
\end{align}
where the last equality follows from the fact that $\left\{\frac{\pi_{j}   v_{j} (\alpha)}{\pi^\top v(\alpha)}: j=1,...,N\right\}$ are convex multipliers.

Finally, we consider the first term in~\eqref{eq-first-deriv-of-g-t}. Expanding the term under the logarithm to complete $Q\left(s^t|s_1\right)$ (see eq~\eqref{eq-Q-s-t}), we obtain
\begin{align}\label{eq-first-term}
& \frac{P^\alpha(s^t) e^{-\alpha (1-\alpha) \frac{\beta^2(s^{t-1})}{2}} v_{s_t}(\alpha) \log P(s^t)}
 {\lambda(\alpha)^{t-1} \pi^\top v(\alpha)} \nonumber \\
 & = \frac{1}{\alpha}\frac{\pi_{s_1} v_{s_1} (\alpha)}{ \pi^\top v(\alpha)} Q\left(s^t|s_1\right)
 \log \left(\frac{\pi_{s_1}  Q\left(s^t|s_1\right) \lambda(\alpha)^{t-1} v_{s_1}(\alpha)}
 {e^{-\alpha (1-\alpha) \frac{\beta^2(s^{t-1})}{2}} v_{s_t}(\alpha)} \right) \nonumber \\
 & =  \frac{\pi_{s_1} v_{s_1} (\alpha)}{\pi^\top v(\alpha)}
 \left( \frac{1}{\alpha} Q\left(s^t|s_1\right)\log \pi_{s_1}\frac{v_{s_1}(\alpha)}{v_{s_t}(\alpha)} +
 \frac{1}{\alpha}  Q\left(s^t|s_1\right)\log Q\left(s^t|s_1\right) \right.\nonumber\\
& \phantom{=} \left.+ \frac{t-1}{\alpha} Q\left(s^t|s_1\right) \log \lambda(\alpha) + (1-\alpha)Q\left(s^t|s_1\right)  \frac{\beta^2(s^{t-1})}{2} \right).
\end{align}
Since $\pi_{s_1}\frac{v_{s_1}(\alpha)}{v_{s_t}(\alpha)}$ is bounded, it is easy to see that
\begin{equation}\label{eq-first-term-first}
\lim_{t\rightarrow +\infty} \frac{1}{t\alpha} \sum_{s^t} Q\left(s^t|s_1\right)
\log \pi_{s_1}\frac{v_{s_1}(\alpha)}{v_{s_t}(\alpha)}=0.
\end{equation}
Further, using AEP~\cite{Cover06}, it can be shown that for any $s_1$,
\begin{equation}\label{eq-first-term-second}
\lim_{t\rightarrow +\infty} \frac{1}{t} \sum_{s^t} Q\left(s^t|s_1\right)
\log Q\left(s^t|s_1\right) = -H(\alpha).
\end{equation}
As for the limit corresponding to the last term in~\eqref{eq-first-term}, we use~\eqref{eq-limit-sum-Q-beta}.
Summarizing~\eqref{eq-second-term-final},\eqref{eq-first-term-first},\eqref{eq-first-term-second},
\eqref{eq-limit-sum-Q-beta} yields that the sequence of first order derivatives of $g_t$ is pointwise convergent with the following limit:
\begin{align}\label{eq-first-order-der-g-t-limit}
\lim_{t\rightarrow +\infty} \frac{d}{d \alpha} g_t(\alpha)& = (2\alpha-1) R(\alpha) - \frac{1}{\alpha} H(\alpha)
+ \frac{1}{\alpha}\log \lambda(\alpha) + (1-\alpha) R(\alpha)\\
&= \alpha R(\alpha) - \frac{1}{\alpha} H(\alpha) + \frac{1}{\alpha}\log \lambda(\alpha).
\end{align}

We now recall expression~\eqref{eq-g-t-equals}, and recall that, from Lemma~\ref{lemma-convexity-of-log-lambda}, we know that $\lambda$ is differentiable and that each component of $v$ is differentiable. From~\eqref{eq-g-t-equals} it is easy then to show that the sequence $\frac{d}{d\alpha} g_t(\alpha)$ is uniformly differentiable on $\alpha \in [0,1]$. By Theorem 7.17 from~\cite{Rudin76}, we have that, for each $\alpha \in [0,1]$,
\begin{align}\label{eq-first-order-der-g-t-limit-lambda}
\lim_{t\rightarrow +\infty} \frac{d}{d \alpha} g_t(\alpha) & = \frac{d}{d \alpha} \lim_{t\rightarrow +\infty} g_t(\alpha)\nonumber \\
&= \frac{d}{d \alpha} \log \lambda(\alpha).
\end{align}
Multiplying with $\alpha$ in~\eqref{eq-first-order-der-g-t-limit}, and rearranging the terms, we get
\begin{equation*}
H(\alpha) - \alpha^2 R(\alpha)= \log \lambda(\alpha) - \alpha \frac{d}{d\alpha} \log \lambda(\alpha).
\end{equation*}
Computing now the first order derivative on both sides yields
\begin{equation*}
\frac{d}{d\alpha} \left(H(\alpha) - \alpha^2 R(\alpha)\right) = - \frac{d^2}{d\alpha^2} \log \lambda(\alpha).
\end{equation*}
By Lemma~\ref{lemma-convexity-of-log-lambda}, we know that $ \log \lambda(\alpha)$ is convex,
implying that the right hand side of the preceding equation is negative. This completes the proof of part~\ref{part-3-H-and-R-properties}.

\end{proof}

\mypar{Proof of Lemma~\ref{lemma-cvx-general}}
\begin{proof}
We prove that $f$ is convex by showing that its Hessian is a positive semi-definite matrix at every point $x\in \mathbb R^d$. It is easy to see that the gradient of $f$ at $x$, $\nabla f(x)$, is given by
\begin{equation}
\label{eq-grad-f}
\nabla f(x) = - \frac{1}{2}\frac{\sqrt{h(x)}}{\sqrt{g(x)}} \nabla g(x) - \frac{1}{2}\frac{\sqrt{g(x)}}{\sqrt{h(x)}} \nabla h(x),
\end{equation}
where $\nabla g(x)$ and $\nabla h(x)$ denote the gradients of $g$ and $h$, respectively, at $x$.
Further, it can be shown that the Hessian of $f$ at $x$, $\nabla^2 f(x)$, equals
\begin{align}
\label{eq-Hessian-f}
\nabla^2 f(x) & = - \frac{1}{2}\frac{\sqrt{h(x)}}{\sqrt{g(x)}} \nabla^2 g(x)
- \frac{1}{2}\frac{\sqrt{g(x)}}{\sqrt{h(x)}} \nabla^2 h(x) \\
& - \frac{1}{2} \nabla g(x) \frac{ \frac{1}{2} \frac{\sqrt{g(x)}}{\sqrt{h(x)}} \nabla^\top h(x)
 - \frac{1}{2} \frac{\sqrt{h(x)}}{\sqrt{g(x)}} \nabla^\top g(x)  }{g(x)}\nonumber \\
& - \frac{1}{2} \nabla h(x) \frac{ \frac{1}{2} \frac{\sqrt{h(x)}}{\sqrt{g(x)}} \nabla^\top g(x)
- \frac{1}{2} \frac{\sqrt{g(x)}}{\sqrt{h(x)}} \nabla^\top h(x)  }{h(x)}.
\end{align}
Since $g$ and $h$ are concave, the first two terms in~\eqref{eq-Hessian-f} are positive semi-definite matrices. Rearranging the remaining two terms, we obtain that the sum of the last two terms in~\eqref{eq-Hessian-f} equals
\begin{align*}
& \frac{1}{4} \frac{1}{g^{\frac{3}{2}}(x) h^{\frac{3}{2}}(x)} \left( h^2(x) \nabla g(x) \nabla g(x)^\top - g(x)h(x)\nabla g(x) \nabla h(x)^\top   \right.\\
&\phantom{\frac{1}{4} \frac{1}{g^{\frac{3}{2}}(x) h^{\frac{3}{2}}(x)}} \left.+\, g^2(x) \nabla h(x) \nabla h(x)^\top - g(x)h(x)\nabla h(x) \nabla h(x)^\top \right)\nonumber\\
& = \frac{1}{4} \frac{1}{g^{\frac{3}{2}}(x) h^{\frac{3}{2}}(x)} \left( h(x) \nabla g(x) - g(x) \nabla h(x) \right) \left( h(x) \nabla g(x) - g(x) \nabla h(x) \right)^\top,
\end{align*}
which is also a positive semi-definite matrix. Since $x$ was arbitrary, this proves that $f$ is convex.
\end{proof}

\bibliographystyle{IEEEtran}
\bibliography{IEEEabrv,BibliographyRandom_walk}
\end{document}